\documentclass[11pt]{article}

\usepackage[utf8]{inputenc}
\usepackage[T1]{fontenc}
\usepackage{amsmath,amssymb,amsfonts,amsthm}
\usepackage{booktabs}
\usepackage{authblk}
\usepackage[margin=1in]{geometry}
\usepackage{enumitem}
\usepackage{array}
\usepackage[hidelinks]{hyperref}
\hypersetup{pdftitle={Bidirectionality is Necessary in Cellular Machines for Exact-Arrival Multi-Agent Path Finding},
  pdfauthor={Masashi Aono},
  pdfkeywords={multi-agent path finding; exact arrival; cellular machine; space-time DAG;
  path-support; direction lower bound; one-way cellular automata; global aggregate flags;
  transcript counting; direct sum; machine-checked proof; Lean}}

\newtheorem{theorem}{Theorem}
\newtheorem{proposition}{Proposition}
\newtheorem{corollary}{Corollary}
\newtheorem{lemma}{Lemma}
\newtheorem{definition}{Definition}
\newtheorem{conjecture}{Conjecture}
\newtheorem{remark}{Remark}

\newcommand{\pred}{\mathrm{pred}}
\newcommand{\suc}{\mathrm{succ}}

\newcommand{\PRED}{\textnormal{\textsc{Pred}}}
\newcommand{\SUCC}{\textnormal{\textsc{Succ}}}
\newcommand{\Cav}{\mathcal{C}}
\newcommand{\CEA}{\mathrm{CellularEA}}
\newcommand{\Cone}{\mathrm{Cone}}
\newcommand{\Lab}{\mathrm{ConeLabels}}

\title{\textbf{Bidirectionality is Necessary in Cellular Machines\\
for Exact-Arrival Multi-Agent Path Finding}}

\author[1,2]{Masashi Aono\thanks{Corresponding author: \texttt{aono@amoebaenergy.com}}}
\affil[1]{Amoeba Energy Co., Ltd., Tokyo, Japan}
\affil[2]{Keio University, Fujisawa, Kanagawa, Japan}

\date{Preprint --- v0.3 / 2026-07-21}

\begin{document}
\maketitle

\begin{abstract}
On ranked directed acyclic graphs coupled by vertex conflicts we study \emph{exact-arrival} multi-agent path
finding---departure and arrival pinned in place and time---posed as a distributed output task: every cell must
emit its own bit of a valid joint path-support mask. The machines are \emph{cellular}: one automaton per cell,
state factored over cells, communication only along a fixed local fabric, synchronous rounds, a
query-independent output round, no central decoder, and a per-topology transition code. ``Necessary'' here
means: on an explicit hard topology, over all its feasible queries, no universally correct \emph{local-only}
architecture can renounce either temporal orientation; with global flags, impossibility becomes a
bandwidth--round trade-off. The class is not vacuous---a flooding machine solves every instance in as many
rounds as the diameter of the local fabric---and the bound survives two concessions: we supply every agent's
exact individual cavity, the endpoint-conditioned reachability cache, which alone solves the single-agent
problem in zero rounds and makes the single-agent direction theorem vacuous; and we allow $k$ one-bit global
aggregate flags. On a twelve-cell, two-agent, depth-three topology no forward-only machine with $k=0$ is
correct on two feasible queries, for any number of rounds; the time-reversed topology does the same to
backward-only machines. Placing $p$ copies side by side gives, for forward-only machines, $J_R+kR\ge p$ at
\emph{constant depth}, where $J_R$ measures how much of the query reaches the targets' cones and is $0$ here. A
backward-only machine solves that family in constant rounds and a flag-streaming one in
$\lceil p/k\rceil+O(1)$, so the bound is tight in order; on the composite family neither orientation can be
renounced. The finite lemmas of the gadget and the pair-invariance contradictions are checked in Lean~4; the
semantic bridge, the cone lemma and the counting are paper proofs. We report no benchmarks.
\end{abstract}

\medskip
\begin{sloppypar}
\noindent\textbf{Keywords:} multi-agent path finding; exact arrival; cellular machine; space-time DAG;
path-support; direction lower bound; one-way cellular automata; global aggregate flags; transcript counting;
direct sum; machine-checked proof.
\end{sloppypar}

\section{Introduction}
\label{sec:intro}

\subsection{The obstruction}
\label{sec:obstruction}

Fix a routing problem in which every agent must leave at a given time from a given place and \emph{arrive at a
given time at a given place}. Call this the \emph{exact-arrival} (EA) discipline. It is the regime of extended
multi-agent path finding with pinned departure and arrival windows \cite{stern2019,mapftw,precedence}.

Consider a cell early in time that must decide whether it belongs to its agent's trajectory. A branch that is
perfectly legal for that agent \emph{in isolation} may be the wrong branch, because several time steps later it
collides with another agent whose own set of legal options is fixed by \emph{that} agent's endpoints. The early
cell's correct output is then a function of information that lies both in its own strict future and in a
different agent. We call this the \emph{cross-agent future-dependence obstruction}, and the paper is an attempt
to say precisely what it costs a cellular machine to face it.

We study the obstruction inside the EA discipline, where it is operationally natural: an agent that commits
early to a branch often cannot repair the commitment later, having no slack. We do \emph{not} claim that
exactness is the cause. The proofs below never use the fact that early arrival is illegal, and we show no
separation between EA and the free-arrival or deadline variants (Section~\ref{sec:nonclaims},
item~\ref{nc:ea}).

\subsection{Why the machine class is cellular}
\label{sec:whycellular}

The class we study is narrow. One finite automaton per cell; the global state is a product over cells; messages
travel only along the local fabric; updates are synchronous rounds. We call such a machine a \emph{cellular
machine}.

The class is \emph{independently motivated by hardware locality, and the theorem is conditional on that
independently motivated choice}. The motivation has nothing to do with the present result: domain-specific
architectures are a standard response to the slowdown of general-purpose scaling \cite{hennessy}; wiring cost
and delay grow with distance, so one wires neighbours; a design that is one tile repeated is what can be laid
out, clocked and yield-tested at scale; power is dominated by long interconnect rather than by local logic; and
the problem itself is local, since an agent's legality constraints involve its own neighbourhood in space and
time. Locally coupled arrays---optionally augmented with one or a few global broadcast or reduction lines---are
a classical architecture pattern for exactly these reasons \cite{stout1983,bokhari1984,chen1990}, and the
amoeba-inspired combinatorial-optimization machines that motivate this work sit in the same pattern
\cite{aono2020jjap,okuyama2023,aono2025jrsj,mohseni2022}.

The logic is therefore additive, not circular: locality is justified first, on independent grounds; the theorem
then adds one further requirement \emph{inside} that class, namely that the local fabric carry causally
effective dependence in both temporal orientations. A designer who has already decided to build a cellular
fabric still has a free choice about whether its causal dependence runs only from earlier cells to later ones,
and that remaining choice is what we show is not free. It should be said plainly that the results depend
essentially on locality and on factored state; they are not statements about unrestricted solvers
(Proposition~\ref{prop:jointlift}).

\subsection{What is at stake is one round of wall-clock time}
\label{sec:stakes}

We make no claim about complexity classes. Nothing here says that bidirectionality evades NP-hardness or
exponential time. The quantity at stake is the \emph{wall-clock duration of one update of the whole instance}.
In a cell-parallel fabric---one fixed circuit per cell, bounded degree, bounded per-cell state and messages,
one hop per unit of physical time---the duration of a round does not grow with the instance, whereas
serializing the same information through a narrow global channel makes the number of rounds grow. The
separation we prove is between $O(1)$ and $\Omega(N/k)$ rounds on a specific family, under assumptions listed
in full in Section~\ref{sec:hwassumptions}.

\subsection{What ``necessary'' means in the title}
\label{sec:necessary}

The title names a necessity, and the necessity we prove is triply limited. We state the limits here rather
than leaving them to a later section.

\begin{enumerate}[leftmargin=1.8em,itemsep=2pt,label=(N\arabic*)]
\item \textbf{The task is limited.} We ask for per-cell output of \emph{some} valid joint path-support mask
under a \emph{feasibility promise}. Detecting infeasible queries, returning a certificate of infeasibility, or
being complete for $\mathrm{FAIL}$ is neither required nor studied
(Section~\ref{sec:nonclaims}, item~\ref{nc:feasible}).
\item \textbf{The machine is limited.} Deterministic, synchronous, one hop per round, factored state, a
transition code that may be \emph{nonuniform}---compiled per topology---a common output round independent of
the query, and no central decoder.
\item \textbf{The necessity is limited.} It is \emph{architecture-level} and \emph{query-family-level}: there
is an explicit hard topology such that, over all of its feasible queries, no universally correct
\emph{local-only} architecture is forward-only and none is backward-only. We do \emph{not} show that within a
single query one causal path must traverse both orientations. With global flags the impossibility is replaced
by a bandwidth--round trade-off, and with unbounded rounds and one flag a forward-only machine exists
(Section~\ref{sec:nonclaims}, item~\ref{nc:bus}).
\end{enumerate}

\subsection{Contributions}
\label{sec:contrib}

\begin{enumerate}[leftmargin=1.6em,itemsep=3pt]
\item A single machine-class definition, $\CEA(k)$ (Section~\ref{sec:class}), closing every assumption in one
place, with round semantics fixed down to send and receive phases and directionality defined by \emph{actual
dependence}.
\item \textbf{Proposition~\ref{prop:nonvacuous}}: the class is not vacuous. A flooding machine solves every
instance and every feasible query within $D$ rounds, $D$ the undirected diameter of the local fabric, and on
the composite family of Corollary~\ref{cor:composite} it is necessarily bidirectional. The lower bounds therefore separate directions inside
a class that can actually solve the problem.
\item The observation that the individual cavity is the maximum individual path-support
(Lemma~\ref{lem:cavity-max}), so handing it over trivializes the single-agent problem and makes the
single-agent direction theorem vacuous (Corollary~\ref{cor:t1-vacuous}). Any necessity claim surviving endpoint
caching must be genuinely multi-agent.
\item \textbf{Theorem~\ref{thm:A}}: a twelve-cell, two-agent, depth-three EA topology and two feasible queries
on which \emph{no} forward-only machine in $\CEA(0)$ is correct, with unbounded rounds allowed and every
agent's exact cavity supplied---plus an explicit time-reversed topology for backward-only machines and a
composite family excluding both at once (Corollary~\ref{cor:composite}).
\item \textbf{Theorem~\ref{thm:Bprime}}: for forward-only machines, the direct-sum trade-off $J_R+kR\ge p$ at
\emph{constant depth}, with $J_R$ defined as a property of the topology and query family rather than as a
resource of the machine
(Definition~\ref{def:JR}). Section~\ref{sec:nosubdiv} explains why the subdivision a distance-based argument
needs is unnecessary here, and why substituting our family into a distance-based bound yields nothing.
\item Matching upper bounds: $\lceil p/k\rceil+O(1)$ rounds by streaming through the flags
(Proposition~\ref{prop:U1}), and a constant number of rounds with no flags at all
(Proposition~\ref{prop:U2}). The latter machine is \emph{backward-only}, not bidirectional
(Remark~\ref{rem:U2}); the bidirectional constant-round upper bound is stated on the composite family
(Corollary~\ref{cor:separation}), and keeping those two apart is what makes the title's claim the right one.
\item A second obstruction (Section~\ref{sec:baseline}), carried out in the strict arrival discipline of the
theorems and then repeated, as an independent study, in the more permissive discipline that lets a vehicle
arrive early and wait. Its evidence comes in four grades, kept apart throughout: the basic families are
\emph{exhaustively enumerated}; the \emph{feasibility} of a parameterized family is shown by an \emph{explicit
verified witness}; the \emph{infeasibility} of one is a \emph{paper argument} spot-checked for small
parameters; and the policy results are an \emph{exhaustive simulation over a finite grid} of departure rules,
push rules and tie orders. On that basis, for $r=1,\dots,8$ the rule residual-$r$ is defeated under every push
rule and tie order in the grid; the uniform statement for all $r$ is a conjecture
(Conjecture~\ref{conj:coverage}); and earliest departure is defeated only for some settings of the remaining
parameters. The instance generators, the policy semantics, the enumerator and a feasibility matrix over both
semantics are shipped as an ancillary program with frozen expected output.
\item An honest verification ledger (Section~\ref{sec:status}) separating mathematical status, mechanization
status and priority status.
\end{enumerate}

\subsection{Relation to the companion paper}
\label{sec:companion}

The equilibrium half of this programme is a companion paper \cite{paperS}, which proves that for a concrete
four-rule family on space-time DAGs, theoretical stability holds exactly when the occupied set is a
path-support subgraph between the pinned endpoints, and that the inclusion-minimal sufficient rule bases are
exactly three, each pairing a predecessor obligation \PRED\ with a successor obligation \SUCC. Two of its results
start us off. Its limits section shows that under vertex exclusion safe supports are not closed under
union---so there is no canonical maximal joint object---and that independent any-choice extraction does not
tensorize, the joint solution relation being a proper subset of the product of its per-agent marginals. Those
two failures are the semantic reason the obstruction of Section~\ref{sec:obstruction} exists: joint extendability
does not factor over agents, though individual reachability does. The present paper prices the failure. We
also use the companion's definitions of path-support and of \PRED/\SUCC.

One caution about the relationship. \emph{The equilibrium obligations \PRED\ and \SUCC\ are not temporal
communication directions, and the companion theorem does not imply the present lower bound; the two results are
counterparts only at the level of two-sided endpoint dependence.} Indeed the companion family contains minimal
bases whose rules all act in one temporal direction while still imposing both obligations. The companion's
counterexamples motivate; the theorem below is proved independently.

\section{Preliminaries}
\label{sec:prelim}

\subsection{Instances}
\label{sec:instances}

\begin{definition}[Ranked DAG, agents, slots]
\label{def:instance}
A \emph{ranked DAG} is a finite directed graph $D=(Q,E)$ with a time function $\tau:Q\to\mathbb{N}$ such that
$(u,v)\in E\Rightarrow\tau(v)=\tau(u)+1$. Read $\tau(q)$ as the scheduled-time layer of cell $q$; it is a
coordinate of the \emph{instance}, unrelated to the machine's own running time, counted in synchronous rounds.
Write $\pred(q)$, $\suc(q)$ for predecessors and successors and $r\leadsto q$ for a directed path (possibly of
length $0$).

An \emph{instance} consists of $a$ agents, agent $i$ carrying a ranked DAG $D_i=(Q_i,E_i,\tau)$ with the $Q_i$
pairwise disjoint, together with a \emph{slot map} $\sigma:Q\to\mathrm{Slots}$ on $Q=\bigsqcup_iQ_i$ that
respects time ($\sigma(q)=\sigma(q')\Rightarrow\tau(q)=\tau(q')$). Two cells of \emph{distinct} agents sharing
a slot are \emph{conflict-adjacent}; no conflict adjacency is placed inside one agent. The \emph{local fabric}
is the undirected graph on $Q$ whose edges are the DAG edges and the conflict adjacencies. Write $N=|Q|$.

Throughout, a \emph{cell} is a vertex of some $Q_i$, that is, a (location, time) pair of one agent's
space-time graph, and a \emph{location} is a vertex of the underlying spatial graph. An instance built on
$\ell$ locations with horizon $H$ therefore has at most $a\cdot\ell\cdot(H+1)$ cells.
\end{definition}

A \emph{query} $\theta$ pins, for each agent $i$, a start $s_i(\theta)$ and a goal $g_i(\theta)$ with
$s_i\leadsto g_i$: both endpoints fixed in place \emph{and} in time.

\begin{definition}[Path-support, thin trajectory, joint validity]
\label{def:validity}
For pins $s,g$ in one agent, a set $S\ni s,g$ is a \emph{path-support subgraph} if every $q\in S$ lies on some
directed $s\to g$ path all of whose cells are in $S$. Equivalently, on a finite DAG, $S$ satisfies
\[
\PRED:\ \forall q\in S\setminus\{s\}:\ \pred(q)\cap S\neq\varnothing,
\qquad
\SUCC:\ \forall q\in S\setminus\{g\}:\ \suc(q)\cap S\neq\varnothing
\]
\cite{paperS}. A \emph{thin trajectory} is a single directed $s\to g$ path. A family $(S_i)_i$ is a
\emph{valid joint output} for $\theta$ if each $S_i$ is a path-support for $(s_i,g_i)$ with both pins ON and no
two ON cells of distinct agents share a slot (\emph{slot exclusion}). A query is \emph{feasible} if a valid
joint output exists.
\end{definition}

\subsection{The cavity is the answer for one agent}
\label{sec:cavity}

\begin{definition}[Cavity]
\label{def:cavity}
For agent $i$ under query $\theta$ the \emph{individual cavity} is
$\Cav_i(\theta)=\{q\in Q_i:s_i\leadsto q\}\cap\{q\in Q_i:q\leadsto g_i\}$: the intersection of forward
reachability from $s_i$ with backward reachability to $g_i$, computed inside $D_i$ alone and ignoring all
conflicts.
\end{definition}

\begin{lemma}[The cavity is the maximum individual path-support]
\label{lem:cavity-max}
On a finite DAG with $s\leadsto g$, the set $\Cav=\{q:s\leadsto q\}\cap\{q:q\leadsto g\}$ is an $s$--$g$
path-support, and every $s$--$g$ path-support $S$ satisfies $S\subseteq\Cav$. Hence $\Cav$ is the unique
maximum path-support under inclusion.
\end{lemma}

\begin{proof}
Let $q\in\Cav$. Concatenating a path $s\leadsto q$ with a path $q\leadsto g$ gives a walk through $q$; as $D$
is acyclic the walk is a path. Every cell $r$ on it satisfies $s\leadsto r$ and $r\leadsto g$, so $r\in\Cav$;
thus $q$ lies on an $s\to g$ path inside $\Cav$. Conversely if $q\in S$ then some $s\to g$ path inside $S$
passes through $q$, whence $s\leadsto q$ and $q\leadsto g$.
\end{proof}

\begin{corollary}[Single-agent trivialization]
\label{cor:single-trivial}
If each cell is given its own cavity bit as a runtime input, a single-agent EA instance is solved in
\emph{zero} rounds with \emph{no} communication: every cell outputs its own cavity bit, which by
Lemma~\ref{lem:cavity-max} is a valid path-support mask with both pins ON. If a thin output is demanded, the
any-choice forward extraction of the companion paper \cite{paperS} thins it.
\end{corollary}

The cavity is not a hint about the single-agent problem; it \emph{is} the answer to it. What remains hard once
it is supplied is not individual reachability but \emph{joint extendability}: which of an agent's individually
legal branches can be extended to a goal simultaneously with the other agents'. That is exactly the object the
companion paper shows to be neither union-closed nor a product of its marginals \cite{paperS}.

\subsection{A no-go: unrestricted solvers need no direction restriction}
\label{sec:nogo}

\begin{proposition}[Joint lift; synchronized instances]
\label{prop:jointlift}
Assume all agents share one set of discrete time layers and every edge advances exactly one layer. Form the
\emph{joint-configuration DAG} whose nodes are tuples of pairwise non-conflicting cells at a common time, with
an edge when every coordinate advances along its agent's DAG. Then (i) joint feasibility is decidable by a
single forward sweep; (ii) a witness joint trajectory is obtained by the same sweep with predecessor
bookkeeping; (iii) the joint analogue of the cavity requires one forward and one backward sweep.
\end{proposition}

\begin{proof}
A joint start-tuple-to-goal-tuple path in this graph is exactly a conflict-free joint trajectory. (i) is
reachability of the goal tuple; (ii) records one predecessor per reached node and backtracks; (iii) is the
intersection of forward reachability from the start tuple with backward reachability from the goal tuple.
\end{proof}

So ``every solver of multi-agent EA must be bidirectional'' is \emph{false}: for the feasibility bit a
forward-only computation suffices. The price is a state space exponential in the number of agents, which is
what a cellular machine refuses to pay. A second concession: as a \emph{graph predicate}, cavity membership is
centrally trivial---two reachability queries to any oracle or labeling scheme \cite{cohen2003}---so no direction
claim can attach to the predicate. Whatever necessity exists must live in the \emph{operational} pairing of
runtime endpoints, per-cell distributed output and restricted dissemination. It also matters which output task
is posed: a global feasibility bit, one selected path, the exact maximal support, and \emph{any} valid
path-support mask are four different tasks, and this paper concerns the last.

\section{The machine class}
\label{sec:class}

\begin{definition}[Cellular machines for exact arrival]
\label{def:class}
$\CEA(k)$ consists of the machines satisfying all of the following on a fixed instance
(Definition~\ref{def:instance}).

\smallskip
\noindent\textbf{(O) Offline knowledge.} The topology $(D_i)_i$ and the slot map $\sigma$ are known offline and
may be compiled into the transition rule; in particular a cell knows its own static address. The machine is
\emph{endpoint-uniform}: one transition rule family serves \emph{all} queries on that topology.

\smallskip
\noindent\textbf{(I) Runtime input, per cell.} Cell $q\in Q_i$ receives (i) its endpoint markers
$\texttt{isStart}(q)=[q=s_i]$ and $\texttt{isGoal}(q)=[q=g_i]$; (ii) its own cavity bit
$c(q)=[q\in\Cav_i(\theta)]$; and, in the \emph{generous} variant, (iii) the multiset of cavity bits
$\{c(q'):q'\text{ conflict-adjacent to }q\}$. It receives no conflict-pruned joint cavity, no joint
extendability table, no selected path, no other agent's endpoint markers, and no global description of the
query. Write $\mathrm{lab}_\theta(q)$ for this input together with $q$'s static data, and call it $q$'s
\emph{initial labelled view}.

\smallskip
\noindent\textbf{(S) Factored state.} One automaton per cell; the global state is the product of cell states;
there is no processor per joint tuple. The number of states per cell must be finite for each fixed topology but
is \emph{not} assumed constant---see Section~\ref{sec:nonclaims}, item~\ref{nc:bus}.

\smallskip
\noindent\textbf{(R) Rounds, with phases.} Deterministic and synchronous. Round $t\ge1$ consists of three
phases. \emph{Aggregate}: the flag word $G_t$ is computed from the round-$(t{-}1)$ cell states. \emph{Send}:
each cell computes, from its round-$(t{-}1)$ state and $G_t$, one message per incident channel. \emph{Receive}:
each cell computes its round-$t$ state from its round-$(t{-}1)$ state, $G_t$, and the messages that arrived on
its incident channels in this round. No within-round combinational cascade across cells is permitted: a message
sent in round $t$ can influence a state only from round $t$ on, and can influence a message only from round
$t{+}1$ on.

\smallskip
\noindent\textbf{(L) Local channels.} Messages travel along both temporal orientations of each DAG edge---the
\emph{increasing-time} and \emph{decreasing-time} ports---and across conflict adjacencies. The local message
alphabet must be finite but its width is not otherwise restricted.

\smallskip
\noindent\textbf{(G) Global aggregate flags.} There are $k\ge0$ flags. Each carries \emph{exactly one bit} per
round, computed by a fixed associative and commutative aggregation of the cell states, and is read by every
cell in the same round it is computed. An $R$-round run has transcript $\Gamma_R=(G_1,\dots,G_R)$ of exactly
$kR$ bits. (Admitting a separate round-$0$ aggregate would make it $k(R+1)$; we fix the convention to close the
off-by-one.)

\smallskip
\noindent\textbf{(T) Termination.} All cells fix their output at a common round $R$ depending on the topology
and the machine but \emph{not} on the query. No query-dependent external sampling of intermediate states is
permitted.

\smallskip
\noindent\textbf{(Y) Output.} Each cell outputs one membership bit and \emph{the bits are the answer}: no
central postprocessing, no query-dependent reinterpretation, no policy output. On a feasible query the ON-set
must be a valid joint output (Definition~\ref{def:validity}) according to the machine's declared convention,
path-support or thin. $\mathrm{FAIL}$ is an aggregate read only at final readout, never fed back to the cells,
and is forbidden on feasible queries. Only feasible queries are used below.

\smallskip
\noindent\textbf{(D) Directionality, by actual dependence.} A machine is \emph{forward-only} if its
aggregation, message and state functions and its outputs are independent of all messages arriving on
decreasing-time ports---equivalently, in the round-unrolled dependency graph, in which an edge is drawn only
when changing an input coordinate can change some later state, message or output, no past-directed temporal
edge occurs. \emph{Backward-only} is the dual. Conflict edges and flags are \emph{time-neutral} and are not
subject to the restriction. A dead read of a port on which nothing depends does not make a machine
bidirectional.
\end{definition}

\begin{remark}[Why the definition is this long]
\label{rem:long}
Every clause closes an escape route by which a direction lower bound could be evaded: (S) against joint-state
expansion, (L) and (G) against out-of-band endpoint distribution, (T) against smuggling information through a
query-dependent halting time, (Y) against a central decoder that reinterprets cell bits, and (D) against a
machine that syntactically reads a port it does not use.
\end{remark}

\begin{remark}[This is a nonuniform cellular network, and that is deliberate]
\label{rem:nonuniform}
Clauses (O) and (S) let the transition code and the per-cell state depend on the topology. This is \emph{not}
the classical cellular automaton, in which one finite state set and one local rule are repeated identically
across a uniform array; it is a nonuniform, topology-indexed cellular network. The choice is deliberate and it
strengthens the lower bounds: a machine allowed to compile the whole topology and every cell's address into its
rule is at least as powerful as one restricted to a single repeated rule, so an impossibility proved against
the former covers the latter \emph{for the task posed here}. What does not transfer is anything about the
classical cellular automaton setting, and the reason is not nonuniformity but the task: that setting has no
runtime endpoint query, no reachability cache, no per-cell mask output and no ranked-DAG topology
(Section~\ref{sec:nonclaims}, item~\ref{nc:nonuniform}). It correspondingly weakens the upper bounds---Proposition~\ref{prop:U1} does use static addresses
and a compiled round schedule, so its tightness is a statement inside this formal class and \emph{not} inside
the bounded-state regime of Section~\ref{sec:hwassumptions}. The relation to the
uniform one-way cellular automata literature is taken up in Section~\ref{sec:onewayca}.
\end{remark}

\subsection{The class is not vacuous}
\label{sec:nonvacuous}

Before proving that one-directional machines fail we must know that the class contains machines that succeed;
otherwise the results below would only distinguish degrees of failure.

\begin{proposition}[Non-vacuity: gathering]
\label{prop:nonvacuous}
Let $D$ be the undirected diameter of the local fabric---the maximum, over its connected components, of the
component diameter. For every instance there is a machine in $\CEA(0)$ that, on every feasible query, outputs
a correct valid joint mask at the common round $D$. On the composite family $\mathcal{H}_p$ of
Corollary~\ref{cor:composite} that machine is neither forward-only nor backward-only, since no machine in
$\CEA(0)$ correct on $\mathcal{H}_p$ can be either. (On $G^\star$ alone---the gadget of
Definition~\ref{def:gadget}---only forward-only machines are excluded, and on its time reversal
$\widehat{G^\star}$ of Definition~\ref{def:reversal} only backward-only ones.)
\end{proposition}

\begin{proof}
The machine floods. Each cell starts holding its own static data and its initial labelled view; in every round
each cell sends everything it currently holds on every incident channel and adds to its store everything it
receives. By induction on $t$, after $t$ rounds a cell holds the initial labelled views of every cell within
distance $t$ in the local fabric, so after $D$ rounds it holds those of its entire connected component. The
round $D$ depends on the topology only, as clause (T) requires.

It remains to check that a cell's own component is enough. First, the validity constraint of
Definition~\ref{def:validity} decomposes over components. Every cavity lies inside a single component: if
$q\in\Cav_i$ then $s_i\leadsto q\leadsto g_i$, so $q$, $s_i$ and $g_i$ are joined by DAG edges, which are
fabric edges. By Lemma~\ref{lem:cavity-max} every valid $S_i$ is contained in $\Cav_i$, so a cell outside its
own agent's cavity is OFF in every valid output---and it knows this from its own cavity bit alone, with no
communication. Second, slot exclusion couples only conflict-adjacent cells, and conflict adjacencies are
fabric edges, so it never constrains two distinct components against each other. Hence the set of valid joint
outputs is a Cartesian product over components.

Fix once and for all the total order on cells given by their static addresses. Each cell now computes, by the
same deterministic procedure applied to the data of its own component, the lexicographically first valid joint
output of that component's subproblem, and emits its own bit of it. That subproblem is feasible, being the
restriction of a valid joint output of the whole query, which exists by hypothesis. All cells of a component
agree because they apply the same function to the same data, and across components no agreement is needed
because the constraint does not couple them. Finally, the lexicographically first element of a Cartesian
product is the tuple of the lexicographically first elements of its factors, so the assembled output is the
lexicographically first valid joint output of the whole instance, and in particular is valid.
\end{proof}

Four things must be said about this construction, because they are exactly what the rest of the paper is
about.

\begin{enumerate}[leftmargin=1.6em,itemsep=2pt]
\item \textbf{It is expensive in the right currency.} Per-cell state and message width grow with the instance.
Clauses (S) and (L) permit that, so the machine is in the class; the bounded-state hardware subclass
(H1)--(H7) of Section~\ref{sec:hwassumptions} does \emph{not} permit it. Proposition~\ref{prop:nonvacuous} is
therefore a statement about the formal class, not a hardware proposal.
\item \textbf{It needs both orientations.} Gathering along decreasing-time ports is what carries information
from later cells to earlier ones. A forward-only machine cannot gather in that direction at all, which is the
content of Theorem~\ref{thm:A}. (Flooding is of course \emph{also} bidirectional on instances where one
direction would have sufficed; the proposition's content is that it succeeds everywhere, and
Theorem~\ref{thm:A} supplies the instances where both directions are genuinely required.)
\item \textbf{It puts the three regimes in one picture}, on the hard families of Section~\ref{sec:mainA}.
Bidirectional with no global channel: solvable, in $D$ rounds. Forward-only with no global channel: \emph{not}
solvable, even with cavities supplied and unboundedly many rounds (Theorem~\ref{thm:A}). Forward-only with $k$
flags: solvable, at the price $J_R+kR\ge p$, where $J_R$ is the cone-label complexity of
Definition~\ref{def:JR} (Theorem~\ref{thm:Bprime}).
\item \textbf{It shows the obstruction is the direction.} Since the class does contain a machine solving every
instance, the failure of forward-only machines cannot be blamed on the class being too weak to express a
solution. What they lack is an orientation.
\end{enumerate}

\section{Warm-up: the single-agent theorem, and its collapse}
\label{sec:warmup}

We first record the locality lemma in the form used throughout, then the direction theorem available
\emph{without} cavity input, and then the fact that it does not survive the concession of
Section~\ref{sec:cavity}.

\begin{lemma}[Conditioned locality]
\label{lem:cone}
Fix a machine $M\in\CEA(k)$ and a flag transcript $\gamma=(g_1,\dots,g_R)$. For every cell $z$ and every
$t\le R$, the round-$t$ state of $z$, every message $z$ sends in rounds $\le t$, and every contribution $z$
makes to the flag words $G_1,\dots,G_t$, are all functions of $\gamma_{\le t}$ together with the initial
labelled views $\mathrm{lab}_\theta(y)$ of the cells $y$ in a radius-$t$ neighbourhood of $z$, where the
neighbourhood is
\begin{enumerate}[leftmargin=2.0em,itemsep=1pt,label=(\alph*)]
\item the radius-$t$ \emph{undirected ball} of $z$ in the local fabric, for an arbitrary $M$;
\item if $M$ is forward-only, the radius-$t$ \emph{forward ancestor cone}
$\Cone^{\rightarrow}_t(z)$, consisting of the cells that reach $z$ in at most $t$ steps along increasing-time
edges and conflict adjacencies;
\item if $M$ is backward-only, the radius-$t$ \emph{descendant cone} $\Cone^{\leftarrow}_t(z)$, consisting of
the cells reachable from $z$ in at most $t$ such steps.
\end{enumerate}
The same holds for the round-$R$ output of $z$.
\end{lemma}

\begin{proof}
We prove the three statements by one induction on $t$, carrying the state, the outgoing messages and the flag
contributions together. Write $\Cone_t(z)$ for whichever neighbourhood is at issue, and note that in each of
the three cases $\Cone_t(y)\subseteq\Cone_{t+1}(z)$ whenever $y$ is a causally effective in-neighbour of $z$.

\emph{Base.} At $t=0$ no message has been exchanged and the state of $z$ is exactly $\mathrm{lab}_\theta(z)$,
the radius-$0$ neighbourhood; no message has been sent and no flag word computed.

\emph{Step.} Assume the claim through round $t$. By the \emph{aggregate} phase of (R), $G_{t+1}$ is a fixed
aggregation of round-$t$ states; conditioning on $\gamma$ fixes its value, and by the inductive hypothesis
$z$'s own contribution to it is a function of $\gamma_{\le t}$ and $\Cone_t(z)$, hence of $\gamma_{\le t+1}$
and $\Cone_{t+1}(z)$. By the \emph{send} phase, the messages $z$ emits in round $t{+}1$ are functions of its
round-$t$ state and of $G_{t+1}$, hence of $\gamma_{\le t+1}$ and $\Cone_t(z)\subseteq\Cone_{t+1}(z)$. By the
\emph{receive} phase, the round-$(t{+}1)$ state of $z$ is a function of its round-$t$ state, of $G_{t+1}$, and
of the messages that reached it in round $t{+}1$; each such message was emitted by a causally effective
in-neighbour $y$ and, by the inductive hypothesis applied to the send phase at $y$, is a function of
$\gamma_{\le t+1}$ and $\Cone_t(y)\subseteq\Cone_{t+1}(z)$.

It remains to identify the causally effective in-neighbours in each case. In case (a) no restriction is
imposed, so they lie in the undirected radius-$1$ ball. In case (b), clause (D) removes every dependence on
messages arriving on decreasing-time ports, so the effective in-neighbours are those reached along
increasing-time edges and conflict adjacencies, that is, they lie in $\Cone^{\rightarrow}_1(z)$. Case (c) is
dual. Finally the output of $z$ is a function of its round-$R$ state, so the claim covers outputs.

Note the order of definitions: the machine class is defined first and the locality is \emph{derived};
``forward-only'' is a property of actual dependence, not a stipulation about which functions may be computed.
\end{proof}

\begin{theorem}[Single-agent direction theorem, endpoint markers only]
\label{thm:single}
Consider the variant of $\CEA(0)$ in which clause (I)(ii)--(iii) is deleted, so that the only runtime input is
the endpoint markers. On single-agent EA instances, no endpoint-uniform machine of that class correct on all
queries is forward-only, and none is backward-only.
\end{theorem}

\begin{proof}
\emph{Forward-only.} Take the fork with cells $s,a,b,g_1,g_2$, edges $s\to a\to g_1$ and $s\to b\to g_2$ and
ranks $\tau(s)=0$, $\tau(a)=\tau(b)=1$, $\tau(g_1)=\tau(g_2)=2$; being single-agent there are no conflict
adjacencies. Under $\theta_1=(s,g_1)$ the unique $s\to g_1$ path is $s\,a\,g_1$, so every valid mask---thick or
thin---sets $a$ ON; under $\theta_2=(s,g_2)$, symmetrically, $a$ must be OFF. The radius-$0$ ancestor cone of
$a$ is $\{a\}$ and every positive-radius one is contained in $\{s,a\}$; all these labelled views are identical
under the two queries, which differ only in the \texttt{isGoal} markers on $g_1,g_2$, strictly in $a$'s future.
By Lemma~\ref{lem:cone}(b) with $k=0$ a forward-only machine outputs the same bit at $a$ under both queries and errs on one.
\emph{Backward-only.} The dual fork $s_1\to a\to g$, $s_2\to b\to g$, the queries differing only in the
\texttt{isStart} markers, using Lemma~\ref{lem:cone}(c).
\end{proof}

\begin{corollary}[Theorem~\ref{thm:single} is vacuous once cavities are supplied]
\label{cor:t1-vacuous}
In $\CEA(0)$ as defined in Definition~\ref{def:class}---that is, \emph{with} the cavity bits---the fork of
Theorem~\ref{thm:single} is solved in zero rounds by a forward-only machine: cell $a$'s own cavity bit is $1$
under $\theta_1$ and $0$ under $\theta_2$, and by Corollary~\ref{cor:single-trivial} the machine that outputs
own-cavity-bit everywhere is correct.
\end{corollary}

This is the crux. A single agent's cavity is a computation of exactly the two-sided kind---forward reachability
from $s$ intersected with backward reachability to $g$---already performed by the environment, and once it is
handed over there is nothing left for the machine's own backward channel to do. To rebuild a necessity claim on
top of the cavity, the forcing information must be moved outside the target's own cavity, that is, across a
conflict coupling into another agent. Since a single agent has no conflict couplings, the resulting statement
is \emph{necessarily} multi-agent. Corollary~\ref{cor:t1-vacuous} is confirmed by exhaustive enumeration
(Section~\ref{sec:status}).

\section{The main theorem}
\label{sec:mainA}

\subsection{The gadget}
\label{sec:gadget}

\begin{definition}[The crossed-conflict selector $G^\star$]
\label{def:gadget}
The gadget $G^\star$ has two agents and twelve cells---six on each side---arranged in four rank layers, as
displayed below.
\end{definition}

\begin{center}
\small
\begin{verbatim}
                    rank 0      rank 1       rank 2        rank 3

 agent A             a0 -------> a1 -------> a2^0 -------> g^0
 (selector)                        \-------> a2^1 -------> g^1

 agent B             b0 -------> z^0 ------> b2^0 ------\
 (responder)           \------> z^1 ------> b2^1 -------> h

 slots (vertex conflicts), crossed, at rank 2 only:
        sigma(a2^0) = sigma(b2^1) = X0
        sigma(a2^1) = sigma(b2^0) = X1
 every other cell occupies a slot of its own

 endpoints   A: (a0, g^b)   <-- the only query-dependent datum
             B: (b0, h)     <-- identical in every query
 queries     theta_b for b in {0,1};  target cell  z := z^1  (rank 1)
\end{verbatim}
\end{center}

\paragraph{The subdivided variant.}
For $L\ge1$ let $G^\star(L)$ replace the rank-$1$ cell $a1$ of $A$ by a chain $a1_1\to\cdots\to a1_L$ and each
branch cell $z^r$ of $B$ by a chain $z^r\to z^r_2\to\cdots\to z^r_L$, with the conflicts at rank $L+1$ and the
goals at rank $L+2$. Thus $G^\star(1)=G^\star$ and $G^\star(L)$ has $3L+9$ cells. Subdivision stretches the
\emph{undirected} distance from the query-dependent cells to the target to $L$ while leaving everything else
unchanged; it is used only to make Section~\ref{sec:nosubdiv} quantitative.

\paragraph{$G^\star$ is a genuine MAPF instance.}
The construction is abstract, but it is realizable as a vertex-conflict-only time-expanded instance on a shared
location graph rather than being an artefact of an arbitrary slot map. Take locations
$\ell_{a0},\ell_{a1},\ell^0,\ell^1,\ell_{g^0},\ell_{g^1},\ell_{b0},\ell_{z^0},\ell_{z^1},\ell_h$ with one
common adjacency relation. Each agent's ranked DAG is then the time expansion of the sub-graph of admissible
states that agent may occupy---$A$ traversing $\ell_{a0}\to\ell_{a1}\to\ell^{r}\to\ell_{g^{r}}$ and $B$
traversing $\ell_{b0}\to\ell_{z^{r}}\to\ell^{1-r}\to\ell_h$---embedded in that one shared location graph. The
agents need not have the same admissible sub-graph; what they share is the location graph on which conflicts
are defined. The only locations shared at a common time are
$\ell^0,\ell^1$ at rank~2, and the conflicts they generate are exactly the two crossed pairs above. Adding the
wait arcs of a standard time-expanded encoding creates no new $s\to g$ path in either agent's cavity, because
an agent that waits can no longer occupy its own goal at its pinned arrival rank. The exclusivity primitive is
vertex conflict only; edge conflicts, swaps and capacities are not modelled (Section~\ref{sec:nonclaims},
item~\ref{nc:conflicts}).

\subsection{The forced output}
\label{sec:A1}

\begin{lemma}[Forced output; uniqueness]
\label{lem:A1}
For each $b\in\{0,1\}$ the query $\theta_b$ on $G^\star(L)$ is feasible and the valid joint output is
\emph{unique}, namely
\[
S_A=\{a0,a1_1,\dots,a1_L,a2^b,g^b\},\qquad
S_B=\{b0,z^b,z^b_2,\dots,z^b_L,b2^b,h\}.
\]
In particular $Y(z^1)=b$. The statement holds verbatim under the thick (path-support) and the thin
(single-path) convention.
\end{lemma}

\begin{proof}
\emph{(i) $A$'s mask is forced.} With $g_A=g^b$ the cavity is $\Cav_A=\{a0,a1_1,\dots,a1_L,a2^b,g^b\}$: the
cell $a2^{1-b}$ cannot reach $g^b$, and $g^{1-b}$ is reachable from $a0$ only through $a2^{1-b}$, so neither
lies in $\Cav_A$. By Lemma~\ref{lem:cavity-max} every path-support is contained in $\Cav_A$; but $\Cav_A$ is a
single path and a path-support must contain an $s\to g$ path, so $S_A\supseteq\Cav_A$ and hence $S_A=\Cav_A$.
The thin convention gives the same set.

\emph{(ii) $A$ occupies one crossed slot.} Since $a2^b\in S_A$ and $\sigma(a2^b)=\sigma(b2^{1-b})$, slot
exclusion gives $b2^{1-b}\notin S_B$.

\emph{(iii) $B$'s mask is forced.} $B$'s endpoints are query-independent and both branches join $b0$ to $h$, so
$\Cav_B$ is all of $B$. By Definition~\ref{def:validity} the path-supports of $(b0,h)$ inside $B$ are exactly
three: branch $0$ only, branch $1$ only, and both branches. By (ii) the two containing branch $1-b$ are
excluded, leaving branch $b$ alone.

This is the point at which the \emph{thick} convention is pinned as tightly as the thin one. Suppose $S_B$
contained $z^{1-b}$ but not $b2^{1-b}$. Then along $z^{1-b}\to z^{1-b}_2\to\cdots\to b2^{1-b}$ the last ON cell
would have no ON successor, violating \SUCC. A partially-ON branch is not a path-support. This is precisely the
obstruction of Section~\ref{sec:obstruction}: whether an early cell may be switched ON is decided by whether a
later cell of the same branch is blocked by somebody else.

\emph{(iv) Feasibility.} The pair in (i) and (iii) is a path-support for each agent with both pins ON, and $A$
uses slot $X_b$ while $B$ uses $X_{1-b}$. \emph{(v)} By (iii), $z^1\in S_B\iff b=1$.
\end{proof}

\subsection{Cone invariance}
\label{sec:A2}

\begin{lemma}[Cone invariance]
\label{lem:A2}
Write $z=z^1$ for the target of $G^\star(L)$. Then $\Cone^{\rightarrow}_0(z)=\{z^1\}$ and
$\Cone^{\rightarrow}_t(z)=\{b0,z^1\}$ for every $t\ge1$; in particular the cone is contained in $\{b0,z^1\}$ at
every radius. All runtime inputs of clause (I) on those two cells, \emph{including} the generous slot-local
foreign cavity bits of (I)(iii), are identical under $\theta_0$ and $\theta_1$.
\end{lemma}

\begin{proof}
\emph{The cone.} At radius $0$ it is $\{z^1\}$ by definition. The only increasing-time predecessor of $z^1$ is
$b0$, which has none. The only non-trivial slot values are $X_0,X_1$, carried by $a2^0,a2^1,b2^0,b2^1$ at rank
$L+1$; neither $z^1$ (rank $1$) nor $b0$ (rank $0$) is conflict-adjacent to anything. So the cone reaches
$\{b0,z^1\}$ at radius $1$ and closes there, independently of the radius thereafter.

\emph{The labels.} For $b0$: $\texttt{isStart}=1$, $\texttt{isGoal}=0$, $c(b0)=1$ since $B$'s endpoints are the
same in both queries, and the slot-local foreign multiset is empty. For $z^1$: both markers $0$, $c(z^1)=1$
since $\Cav_B$ is all of $B$, and again an empty foreign multiset. None depends on $b$: the query difference is
confined to the cells $\{a2^0,a2^1,b2^0,b2^1,g^0,g^1\}$ of rank $\ge L+1$, which do not meet the cone.
\end{proof}

\subsection{The cavity-robust impossibility}
\label{sec:thmA}

\begin{theorem}[Cavity-robust direction impossibility]
\label{thm:A}
There is a two-agent, twelve-cell, depth-three, constant-degree vertex-conflict EA topology---namely
$G^\star$---and two \emph{feasible} queries $\theta_0,\theta_1$ on it such that the following holds. Even when
every agent's exact individual cavity is supplied as a per-cell runtime input, and even in the generous variant
in which each cell additionally receives the cavity bits of its conflict neighbours, \emph{no} forward-only
machine in $\CEA(0)$ is correct on both queries. The number of rounds $R$ is unrestricted. Dually, on the
time-reversed topology $\widehat{G^\star}$ of Section~\ref{sec:dual}, no backward-only machine in $\CEA(0)$ is
correct on both queries.
\end{theorem}

\begin{proof}
Suppose a forward-only $M\in\CEA(0)$ is correct on $\theta_0$ and $\theta_1$. With $k=0$ the transcript is the
empty word, so by Lemma~\ref{lem:cone}(b) the round-$R$ output of the target $z^1$ is a function of the initial labelled
views of $\Cone^{\rightarrow}_R(z^1)$. By Lemma~\ref{lem:A2} that cone is contained in $\{b0,z^1\}$ at every radius---%
equal to it at every positive radius, and $\{z^1\}$ at radius $0$---and its labelled view is identical under
the two queries. Therefore $M$ outputs the same bit at $z^1$ under $\theta_0$
and $\theta_1$. But by Lemma~\ref{lem:A1} both queries are feasible, each has a unique valid joint output, and
that output forces $Y(z^1)=b$ under either convention; and by (Y) the cell bits are the answer, with no central
decoder permitted to reinterpret them. So $M$ errs on at least one query. The parameter $R$ occurs nowhere in
the argument, so allowing unboundedly many rounds changes nothing. The backward statement follows from
Lemma~\ref{lem:reversal}.
\end{proof}

Three features are worth isolating.

\begin{enumerate}[leftmargin=1.6em,itemsep=2pt]
\item It is an \emph{impossibility}, not a separation: not a lower bound on $R$ but the non-existence of a
correct machine---inside a class that, by Proposition~\ref{prop:nonvacuous}, does contain machines solving the
task.
\item It is \emph{robust to endpoint caching}. By Corollary~\ref{cor:t1-vacuous} the single-agent theorem
evaporates once cavities are supplied; Theorem~\ref{thm:A} does not.
\item It is \emph{essentially multi-agent}. By Corollary~\ref{cor:single-trivial} the corresponding
single-agent statement is false. Two agents are needed because a conflict coupling is needed, and the layered
placement used here---a free target, then a conflict layer, then a query-dependent goal layer---needs three
layers. We do not claim that depth three is globally minimal over all constructions.
\end{enumerate}

\begin{remark}[Quantification]
\label{rem:quant}
The theorem is of \emph{existential-topology, universal-query} form: there exists a topology such that no
forward-only machine is correct on all of its feasible queries. It does \emph{not} say that forward-only
machines fail on every multi-agent family; see Section~\ref{sec:nonclaims}, item~\ref{nc:notall}.
\end{remark}

\subsection{Why no subdivision is needed}
\label{sec:nosubdiv}

A lower bound of the classical shape---either the machine runs long enough for the changed marker to reach the
target, or it pays global bandwidth---needs the changed markers to be \emph{far} from the target in the
undirected local graph, and buys that distance by subdividing edges. Here that is unnecessary, and the reason
is the mathematical content of the construction.

For a forward-only machine the relevant object is not the undirected ball of Lemma~\ref{lem:cone}(a) but the directed cone
of Lemma~\ref{lem:cone}(b). In $G^\star$ the query difference lies in the target's \emph{strict future}, so no forward
directed path reaches the target at any radius: the relevant directed distance is $\infty$, not merely large.
The distance term therefore disappears and only the direct-sum term of Section~\ref{sec:bandwidth} survives, at
constant depth.

The contrast is sharp in the other direction too. In $G^\star(L)$ the \emph{undirected} distance from the
query-dependent cells to the target is exactly $L$, since $b2^1$ is a successor of $z^1$; at $L=1$ it is $1$. A
distance-based bound of the form $R\ge\min\{D,\lceil p/k\rceil\}$---Proposition~\ref{prop:thm5}
below---instantiated on this family therefore yields $\min\{1,\cdot\}=1$, a trivial statement. Both facts---the
invariance of the directed cone and the growth of the undirected distance with $L$---are confirmed by
enumeration for $L=1,\dots,6$ (Section~\ref{sec:status}).

\subsection{The backward dual and the composite family}
\label{sec:dual}

\begin{definition}[Time reversal]
\label{def:reversal}
$\widehat{G^\star}(L)$ is obtained from $G^\star(L)$ by reversing every edge, replacing $\tau$ by
$(L+2)-\tau$, and exchanging each agent's start and goal. The slot map is unchanged, conflicts being
time-neutral.
\end{definition}

Explicitly, for $L=1$: agent $\widehat A$ runs $g^0\to a2^0\to a1\to a0$ and $g^1\to a2^1\to a1\to a0$ with
start $g^b$---the query-dependent datum---and goal $a0$; agent $\widehat B$ runs $h\to b2^r\to z^r\to b0$ with
start $h$ and goal $b0$; the crossed conflicts sit at rank~1 and the target $\widehat z:=z^1$ at rank~2, so the
query difference lies in the target's strict \emph{past}.

\begin{lemma}[Reversal duality]
\label{lem:reversal}
Time reversal (Definition~\ref{def:reversal}) maps ranked DAGs to ranked DAGs, preserves path-supports and slot
exclusion, exchanges feasible queries with feasible queries, and exchanges forward-only with backward-only
machines. Consequently Lemmas~\ref{lem:A1} and~\ref{lem:A2} hold verbatim on $\widehat{G^\star}(L)$ with the
descendant cone $\{z^1,b0\}$ in place of the ancestor cone, and Theorem~\ref{thm:A} holds for backward-only
machines.
\end{lemma}

\begin{proof}
Reversing all edges and replacing $\tau$ by $(L+2)-\tau$ again satisfies $(u,v)\in E\Rightarrow
\tau(v)=\tau(u)+1$. The reverse of an $s\to g$ path is a $g\to s$ path, so path-supports for $(s,g)$ are
exactly path-supports for $(g,s)$ in the reversed DAG; slot exclusion is a statement about simultaneous
occupancy and is invariant. Reversal exchanges increasing-time with decreasing-time ports, hence forward-only
with backward-only machines by clause (D), and correspondingly parts (b) and (c) of Lemma~\ref{lem:cone}.
\end{proof}

\begin{corollary}[One topology excluding both orientations]
\label{cor:composite}
Let $\mathcal{H}_p$ consist of $p$ slot-disjoint copies of $G^\star$ together with $p$ slot-disjoint copies of
$\widehat{G^\star}$. In $\CEA(0)$, no forward-only machine is correct on all feasible queries of
$\mathcal{H}_p$, and no backward-only machine is. Hence a universally correct cellular architecture for
$\mathcal{H}_p$ must admit causally effective dependence in \emph{both} temporal orientations.
\end{corollary}

\begin{proof}
A forward-only machine already fails on any one $G^\star$ copy by Theorem~\ref{thm:A}, and a backward-only one
on any one $\widehat{G^\star}$ copy by Lemma~\ref{lem:reversal}.
\end{proof}

The corollary is at the level of the architecture and of the query family. It does not assert that on a single
query one causal path traverses both orientations (Section~\ref{sec:nonclaims}, item~\ref{nc:samequery}).

\section{Bandwidth: what a narrow global channel buys}
\label{sec:bandwidth}

Theorem~\ref{thm:A} concerns $k=0$. Once $k\ge1$ the impossibility becomes a price; this section computes it
and shows the computation is tight in order.

\subsection{Warm-up: the direction-unrestricted bound}
\label{sec:unrestricted}

\begin{proposition}[Direction-unrestricted bound]
\label{prop:thm5}
Let $k\ge1$. Take $p$ slot-disjoint selector copies and subdivide temporal edges so that the undirected
distance in the local fabric from the query-dependent cells to each target is at least $D$. Then any machine in
$\CEA(k)$---with no restriction on direction---correct on the $2^p$ queries within $R$ rounds satisfies
$R\ge\min\{D,\lceil p/k\rceil\}$.
\end{proposition}

\begin{proof}
If $R\ge D$ there is nothing to prove. If $R<D$ then, by Lemma~\ref{lem:cone}(a), the round-$R$ output of every target is a
function of $\gamma$ together with the initial labelled views of its undirected radius-$R$ ball; those views
are identical across the family because the query-dependent cells are at distance $\ge D>R$. So the output
vector is a function of $\Gamma_R$ alone. Correctness forces $2^p$ distinct output vectors, hence
$2^p\le2^{kR}$.
\end{proof}

\begin{remark}[Where the target cell must sit]
\label{rem:zplacement}
In a selector construction of this kind it matters, and is easy to leave ambiguous, \emph{which} agent carries
the target cell $z_j$. If $z_j$ is a branch cell of the \emph{selector} agent $A_j$---the agent whose goal the
query bit pins---then supplying cavities destroys the bound: $z_j$'s own cavity bit already equals $b_j$, so the
own-cavity-bit machine of Corollary~\ref{cor:single-trivial} answers in zero rounds. The bound survives cavity
input only when $z_j$ is a branch cell of the \emph{responder} agent $B_j$, whose endpoints are the same in
every query. We therefore fix the responder placement throughout, and Lemma~\ref{lem:A2} is the proof that the
placement does what it should. Both readings have been checked by enumeration (Section~\ref{sec:status}).
\end{remark}

\subsection{The direct sum}
\label{sec:directsum}

Let $\mathcal{G}_p$ consist of $p$ slot-disjoint copies of $G^\star$: $2p$ agents, $N=\Theta(p)$ cells,
constant degree and constant depth. The query family is $\Theta=\{0,1\}^p$, copy $j$ carrying
$g_{A_j}=g^{b_j}$; the target set is $Z=\{z_1,\dots,z_p\}$ with $z_j:=z^1_j$.

\begin{lemma}[Direct-sum forcing]
\label{lem:A1sum}
Every $b\in\{0,1\}^p$ is feasible on $\mathcal{G}_p$, the valid joint output is unique, and the target vector
is $(Y(z_1),\dots,Y(z_p))=b$; so the forced output vector takes $2^p$ distinct values. Moreover each $z_j$ has
forward ancestor cone contained in $\{b0_j,z^1_j\}$ at every radius, and equal to it at every positive radius,
with labels constant across the family.
\end{lemma}

\begin{proof}
Slots are disjoint across copies, so slot exclusion decomposes and joint validity is the conjunction of the
per-copy validities; apply Lemma~\ref{lem:A1} to each copy. For the cones apply Lemma~\ref{lem:A2} to each
copy, noting that distinct copies share neither edges nor slots.
\end{proof}

We now isolate the quantity that measures how much of the query reaches the targets. It is a property of the
topology, the query family, the target set and the radius---\emph{not} a syntactic resource of the machine, and
not a string of advice handed to anybody---so it is defined here rather than in Definition~\ref{def:class}.

\begin{definition}[Cone-label complexity]
\label{def:JR}
Fix a topology, a query family $\Theta$, a target set $Z$, a radius $R$ and a direction
$d\in\{\rightarrow,\leftarrow\}$. For $\theta\in\Theta$ let
\[
\Lab^{d}_{R,Z}(\theta)\ :=\ \bigl(\mathrm{lab}_\theta(y)\bigr)_{\,y\in\bigcup_{z\in Z}\Cone^{d}_R(z)}
\]
be the tuple of initial labelled views over the radius-$R$ cones of the targets, and put
\[
J^{d}_R(\Theta,Z)\ :=\ \bigl\lceil\log_2\bigl|\{\Lab^{d}_{R,Z}(\theta):\theta\in\Theta\}\bigr|\bigr\rceil .
\]
We drop $d$ when the direction is clear from the machine under discussion, and write $J_R$ for
$J^{d}_R(\Theta,Z)$. Since cones grow with $R$, $J_R$ is non-decreasing in $R$; a machine enters this
definition only through its round bound $R$ and its direction restriction.
\end{definition}

\begin{lemma}[Directed transcript counting]
\label{lem:F1prime}
Fix a forward-only machine in $\CEA(k)$ halting at round $R$, a query family $\Theta$ and a target set $Z$.
Then the round-$R$ output vector $Y_Z(\theta)$ takes at most $2^{\,J^{\rightarrow}_R(\Theta,Z)+kR}$ distinct
values. The dual statement holds for backward-only machines with $J^{\leftarrow}_R$.
\end{lemma}

\begin{proof}
By Lemma~\ref{lem:cone}(b), $Y_Z(\theta)$ is a function of the pair $(\Lab^{\rightarrow}_{R,Z}(\theta),\Gamma_R(\theta))$.
The machine is deterministic, so $\theta$ determines that pair; the number of pairs is at most
$2^{J^{\rightarrow}_R}\cdot 2^{kR}$, the transcript having exactly $kR$ bits by convention (G).
\end{proof}

Only the second factor is inherited from the standard method; the first is the contribution of the directed
cone.

\begin{remark}[If one really wants an advice model]
\label{rem:advice}
$J_R$ measures information that the query family \emph{happens} to place inside the targets' cones. If one
additionally hands the machine an external advice string of $a$ bits, chosen as a function of the query and
readable by every cell, the same counting gives $a+J_R+kR\ge\log_2M$, where $M$ is the number of forced output
vectors. We do not pursue that model; we record it because it is the natural home of the objection ``why not
precompute the answer'', whose price is then visible as $a\ge p$.
\end{remark}

\subsection{The bandwidth trade-off}
\label{sec:thmBprime}

\begin{theorem}[Cavity-robust direction--bandwidth trade-off]
\label{thm:Bprime}
For every $p\ge1$ and $k\ge0$ there is a $2p$-agent, \emph{constant-depth}, constant-degree vertex-conflict EA
topology with $N=\Theta(p)$ cells and a family $\Theta$ of $2^p$ feasible queries---namely
$\mathcal{G}_p$---such that any deterministic, synchronous, endpoint-uniform \emph{forward-only} machine in
$\CEA(k)$ producing a correct joint valid-support output on all queries at a common round $R$ satisfies
\[
J^{\rightarrow}_R(\Theta,Z)\ +\ kR\ \ge\ p .
\]
On this family $J^{\rightarrow}_R(\Theta,Z)=0$ for \emph{every} $R$, so the bound reads $kR\ge p$: for $k\ge1$,
$R\ge\lceil p/k\rceil$, and for $k=0$ no forward-only machine exists at any finite $R$. The dual statement
holds for backward-only machines on the time-reversed family.
\end{theorem}

\begin{proof}
By Lemma~\ref{lem:A1sum} correctness forces $2^p$ distinct output vectors, and the targets' radius-$R$ forward
cones carry labels constant across $\Theta$, so $|\{\Lab^{\rightarrow}_{R,Z}(\theta):\theta\in\Theta\}|=1$ and
$J^{\rightarrow}_R=0$ for every $R$. By Lemma~\ref{lem:F1prime} at most $2^{J^{\rightarrow}_R+kR}$ output
vectors are available. Take logarithms. If $k=0$ the inequality reads $0\ge p$, false for $p\ge1$, so no such
machine exists.
\end{proof}

No distance parameter appears, and the depth stays at three.

\subsection{Matching upper bounds}
\label{sec:upper}

\begin{proposition}[Streaming through the flags]
\label{prop:U1}
Let $k\ge1$. On the subclass of $\CEA(k)$ whose one-bit flags are programmable as OR aggregations, a machine
that uses \emph{no local channels at all}---hence is forward-only, and backward-only, in the sense of
(D)---solves $\mathcal{G}_p$ in $R=\lceil p/k\rceil+O(1)$ rounds.
\end{proposition}

\begin{proof}
By (O) each cell knows its static address. Index copies by $j\in\{1,\dots,p\}$, rounds from $1$, and flags by
$f\in\{1,\dots,k\}$; assign copy $j$ the round $r(j)=\lceil j/k\rceil$ and the flag $f(j)=((j-1)\bmod k)+1$, so
that distinct copies sharing a round use distinct flags. At round $r(j)$ the cell $a2^1_j$ contributes its own
cavity bit $c(a2^1_j)=b_j$ to flag $f(j)$ and every other cell contributes $0$; under OR aggregation the flag
word carries $b_j$. \emph{Every} cell of copy $j$'s responder branches---that is $z^0_j,z^1_j$, the subdivision
cells $z^r_{j,i}$, and $b2^0_j,b2^1_j$---reads flag $f(j)$ at round $r(j)$ and latches $b_j$, which it can do
because it knows its own copy index from its static address. At the common round $R=\lceil p/k\rceil+1$ every
cell fixes its output: cells of $A_j$ output their own cavity bits; $b0_j$ and $h_j$ output ON; and each
latching cell of branch $r$ outputs ON exactly when the latched $b_j$ equals $r$. By Lemma~\ref{lem:A1} this is
the unique valid joint output.
\end{proof}

Hence Theorem~\ref{thm:Bprime} is tight in order: $\lceil p/k\rceil\le R\le\lceil p/k\rceil+O(1)$. As
Remark~\ref{rem:nonuniform} notes, this upper bound uses static addresses and a compiled schedule, so its
tightness is a statement inside the formal class, and not inside the bounded-state regime of
Section~\ref{sec:hwassumptions}.

\begin{proposition}[Constant rounds with a past-directed local channel]
\label{prop:U2}
A machine with $k=0$, constant per-cell state and constant message alphabet solves $\mathcal{G}_p$ in two
rounds under the phase semantics of Definition~\ref{def:class}(R)---in any case in a constant number of rounds
independent of $p$.
\end{proposition}

\begin{proof}
Round~1: each cell $a2^r_j$ whose own cavity bit is $1$ sends \texttt{blocked} across its conflict edge, which
is received in the same round by $b2^{1-r}_j$. Round~2: each $b2^r_j$ that did not receive \texttt{blocked}
sends \texttt{unblocked} to its predecessor $z^r_j$ along a decreasing-time port, received in that round. Each
$z^r_j$ outputs ON exactly when it received \texttt{unblocked}; each $b2^r_j$ outputs ON exactly when it did
\emph{not} receive \texttt{blocked}---that is, exactly when it sent \texttt{unblocked}; $b0_j,h_j$ output ON;
cells of $A_j$ output their own cavity bits. By Lemma~\ref{lem:A1} this is the unique valid joint output. All copies
run concurrently, so the round count does not depend on $p$. (With subdivision $z^r_j$ is reached after $L$
further past-directed hops, giving $L+1$ rounds---still independent of $p$.)
\end{proof}

\begin{remark}[The exactness caveat, and why the corollary needs the composite family]
\label{rem:U2}
The machine of Proposition~\ref{prop:U2} uses decreasing-time ports and no increasing-time ports, so under
clause (D) it is \emph{backward-only}, not bidirectional. There is no contradiction with Theorem~\ref{thm:A}:
$\mathcal{G}_p$ is hard for forward-only machines and easy for backward-only ones. It follows that a
bidirectionality corollary \emph{cannot} be stated on $\mathcal{G}_p$ alone. It must be stated on the composite
family, where the mirrored copies make each orientation indispensable for a different part of one topology.
\end{remark}

\begin{corollary}[Constant versus linear rounds]
\label{cor:separation}
On the family $\mathcal{H}_p$ of Corollary~\ref{cor:composite}---$4p$ agents, $N=\Theta(p)$ cells, constant
degree and constant depth---every forward-only machine in $\CEA(k)$ satisfies $kR\ge p$, every backward-only
machine satisfies $kR\ge p$, and a machine admitting both orientations solves it with $k=0$ in $R=O(1)$
rounds. Consequently, for $k=O(1)$, no architecture handling all of $\mathcal{H}_p$ in $O(1)$ rounds is
forward-only or backward-only.
\end{corollary}

\begin{proof}
For the forward-only bound, restrict attention to the sub-family $\Theta'\subseteq\Theta$ of queries that fix
the endpoints of every reversed copy and vary only the $\mathcal{G}_p$ part. On $\Theta'$ the targets
$Z=\{z_1,\dots,z_p\}$ of the $\mathcal{G}_p$ part still have $2^p$ forced output vectors by
Lemma~\ref{lem:A1sum}, and their forward cones are unchanged---the reversed copies share neither edges nor
slots with them---so $J^{\rightarrow}_R(\Theta',Z)=0$ and the counting of Theorem~\ref{thm:Bprime} applies
verbatim, giving $kR\ge p$. A machine correct on $\Theta$ is correct on $\Theta'$, so the bound holds on
$\mathcal{H}_p$. The backward-only bound is the mirror argument on the reversed part. For the upper bound run
the machine of Proposition~\ref{prop:U2} on the $\mathcal{G}_p$ part and its time-reversal on the reversed
part, concurrently; the combined machine uses decreasing-time ports on one part and increasing-time ports on
the other, so it is neither forward-only nor backward-only.
\end{proof}

\section{A second obstruction: static offsets in a reactive shortest-path policy}
\label{sec:baseline}

Theorem~\ref{thm:A} is about information that \emph{cannot arrive}. There is a second, logically independent
difficulty in exact-arrival routing, and we record it because it disposes of the most natural engineering
response to Section~\ref{sec:obstruction}---``start everybody earlier'', or ``start everybody later''. Here the
information can arrive; it is a fixed policy that uses it in the wrong direction.

Two things should be said before the details, because they set what this section is and is not.
\emph{The evidence is of a different type from the rest of the paper}: Sections~\ref{sec:mainA}
and~\ref{sec:bandwidth} prove theorems about a machine class, whereas this section reports finite enumerations
and simulations. And \emph{the object under test is a whole policy}, not a departure rule in isolation: the
shortest-path plan, the static priority, the holding convention and the one-step give-way all participate, and
Section~\ref{sec:policyscope} shows concretely that they do. The section title names the policy for that
reason.

Almost everything in this section is established by the accompanying program rather than by a proof, and
nothing in it is machine-checked; the one exception is noted below. We are precise about what kind of evidence
each statement carries. \emph{Basic families are
exhaustively enumerated}: every valid joint schedule is generated. \emph{Feasibility of a parameterized family
is established by an explicit witness schedule} verified against the model. \emph{Infeasibility of a
parameterized family is a paper argument}---the exception just mentioned: a forced collision at $T_B-1$,
checked by the program only for the small values of $r$ it can enumerate. And \emph{policy results are an exhaustive simulation over a finite
grid} of departure rules, push rules and tie orders. We do not enumerate all tie orders, and we do not
enumerate the full schedule set of any parameterized family used below---the one place where a full count of a
parameterized instance is quoted, in Section~\ref{sec:strictfamilies}, is a variant that the section does
\emph{not} use, and the accompanying program prints those counts.

\subsection{Two arrival semantics, and why the distinction is load-bearing}
\label{sec:twosemantics}

An instance is a finite undirected location graph $(V,E)$ with agents $i$ carrying a start $s_i$, a goal $g_i$
and an arrival time $T_i$, and horizon $H=\max_iT_i$. A joint schedule assigns to each agent a sequence
$\pi_i(0),\dots,\pi_i(H)$ of locations with
\begin{enumerate}[leftmargin=2.2em,itemsep=1pt,label=(M\arabic*)]
\item $\pi_i(0)=s_i$;
\item $\pi_i(t{+}1)=\pi_i(t)$ or $(\pi_i(t),\pi_i(t{+}1))\in E$, so waiting is allowed;
\item arrival, in one of the two variants below;
\item \emph{vertex conflict}: $\pi_i(t)\neq\pi_j(t)$ for $i\neq j$;
\item \emph{swap prohibition}: not $\pi_i(t)=\pi_j(t{+}1)$ and $\pi_j(t)=\pi_i(t{+}1)$ when
$\pi_i(t)\neq\pi_j(t)$.
\end{enumerate}

\begin{description}[leftmargin=2.2em,itemsep=3pt]
\item[\textsc{strict}] $\pi_i(t)\neq g_i$ for every $t<T_i$, and $\pi_i(t)=g_i$ for every $t\ge T_i$: the goal
is occupied for the first time exactly at $T_i$. This is the location-graph reading of the discipline used in
the rest of the paper, where the goal is a \emph{cell} at a fixed rank and occupying it early is not
expressible.
\item[\textsc{dwell}] $\pi_i(T_i)=g_i$ and $\pi_i(t)=g_i$ for $t\ge T_i$, with no constraint before $T_i$:
early visits to one's own goal are legal, including leaving and returning. This is \emph{fixed-time goal
occupancy with early goal visits permitted and post-arrival parking}; it is not the standard
deadline-with-dwell condition, and it is strictly more permissive than \textsc{strict}.
\end{description}

The distinction is not cosmetic, and the honest way to say why is to report what happens when one ignores it.
An earlier version of this section used only \textsc{dwell} and asserted that a failure there is a fortiori a
failure under the stricter reading. \emph{That inference is invalid}, and the families it was applied to make
the point concretely: three of them are \emph{infeasible} under \textsc{strict}, so every policy fails on them
there and nothing about departure rules follows. The accompanying program therefore prints a feasibility
matrix, over both semantics and every family, before any policy is run.

\begin{center}
\small
\begin{tabular}{@{}lll@{}}
\toprule
family & \textsc{strict} & \textsc{dwell} \\
\midrule
$S_{\mathrm{hold}}$ & 9 schedules & 10 schedules \\
$S_{\mathrm{late}}$ & 11 schedules & 34 schedules \\
$S_{\mathrm{late}}(r)$, $r=1,2,3$ & feasible (witness) & feasible (witness) \\
\addlinespace
$G_{\mathrm{give}}$ & 9 schedules & 10 schedules \\
$G_{\mathrm{late}}$ & \textbf{infeasible} & 11 schedules \\
$G_{\mathrm{late}}$-sym & \textbf{infeasible} & 48 schedules \\
$G_{\mathrm{late}}(r)$, $r=1,2,3$ & \textbf{infeasible} & feasible (witness) \\
\bottomrule
\end{tabular}
\end{center}

Three of the four $G$ families collapse---$G_{\mathrm{give}}$ does not---and the reason is worth stating
precisely, because it is what the late families had to repair. (The families themselves, and the symbols
$x_i$, the responder and the crossers $A_i$ used here, are defined in Sections~\ref{sec:strictfamilies}
and~\ref{sec:dwellstudy}; this paragraph is a forward look at why the matrix reads as it does.) The common
cause is that a goal is \emph{adjacent} to a contested location, which leaves nowhere to redistribute waiting
to; but the collision partner differs. In $G_{\mathrm{late}}$ and its symmetric variant \emph{both} goals are
adjacent to the single shared location $x$ and the two arrival times are equal, so both agents must occupy $x$
at $T-1$ and collide with \emph{each other}; neither has zero slack, and there is no third agent. In
$G_{\mathrm{late}}(r)$ the responder must occupy $x_{r+1}$ at its own $T-1$, where the zero-slack crosser
$A_{r+1}$ also is.

Accordingly this section is in two parts. Section~\ref{sec:strictfamilies} works in \textsc{strict}, the
semantics of the rest of the paper. One of its families, $S_{\mathrm{hold}}$, is literally the same instance as
$G_{\mathrm{give}}$---same graph, same agents, same arrival times---read under the other semantics; it needed
no repair, being already feasible under \textsc{strict}, and what changes is only where an early-departing
agent has to put its waiting. The two late families are new, built with a private corridor after the contested
location so that waiting can be redistributed there. Section~\ref{sec:dwellstudy}
keeps the \textsc{dwell} results as an independent robustness study of a different, and in practice common,
arrival discipline---one in which a vehicle may arrive early and wait on its stand. \emph{We do not transport
the results of Section~\ref{sec:dwellstudy} to \textsc{strict}}, and the main theorems of this paper concern
\textsc{strict} only.

\subsection{The policy class}
\label{sec:policyclass}

A policy $\Pi(\delta,\mathrm{push},\mathrm{tie})$ is fixed by three parameters. Writing
$\mathrm{slack}_i=T_i-\mathrm{dist}(s_i,g_i)$:

\begin{itemize}[leftmargin=1.6em,itemsep=2pt]
\item the \textbf{departure rule} $\delta$ gives the number of steps agent $i$ waits at its start before
following its plan: $0$ for \emph{earliest}, $\mathrm{slack}_i$ for \emph{latest}, and
$\max(0,\mathrm{slack}_i-r)$ for \emph{residual-$r$}, which keeps $r$ units of slack in hand. Thus latest $=$
residual-$0$, and earliest $=$ residual-$r$ for any $r$ at least the largest slack;
\item the \textbf{push rule} says who may be pushed aside by a higher-priority mover: under \textsc{yw} only
an agent camping on its own goal ahead of time; under \textsc{yh} only an agent that is stationary
\emph{by plan}; under \textsc{ya} any stationary agent. Under \textsc{strict} no agent ever occupies its goal
early, so \textsc{yw} is \emph{vacuous} there and \textsc{yh} is its natural analogue;
\item the \textbf{tie-break} is a total order on locations, used among give-way targets. The sweeps below
range over \emph{two chosen} orders per family---the one preferring the \emph{corner} (the far, dead-end
location of the corridor, $c_5$ in $S_{\mathrm{hold}}$) and the one preferring the pocket $w$, or ascending
and descending location index where those two roles are absent---never over all total orders. ``All six
push/tie combinations'' below therefore means the three push rules against those two orders.
\end{itemize}

Each agent takes a free-arrival shortest path to its goal ignoring the others, ties by ascending location
index. Under \textsc{dwell} it walks to the goal and camps. Under \textsc{strict} it cannot camp on its goal, so it must
wait somewhere else, and the policy has it walk the path, stop at the \emph{penultimate} location, hold there
until $T_i-1$ and enter the goal at $T_i$. \emph{This is a choice made by the policy, not a consequence of the
semantics}: holding at the last location before the goal is what commits latest, but the strict discipline
merely forbids early goal occupancy and says nothing about where the waiting goes. Section~\ref{sec:policyscope}
records what that choice costs the interpretation of the results. Priority is static:
ascending slack, then ascending agent index. Execution is synchronous, and within a step the proposals are
resolved to a fixed point by (i) withdrawing swaps, lower priority yielding; (ii) awarding a contested
location to the highest priority; and (iii) \emph{push-aside}: a mover whose target is held by a stationary
lower-priority agent that $\mathrm{push}$ permits to be moved makes that agent step to one legal adjacent
location---excluding locations proposed by others, locations held by other stationary agents, swaps, and under
\textsc{strict} the pushed agent's own goal before its arrival time---preferring the smallest free-arrival
distance to its own goal and breaking ties by $\mathrm{tie}$; the pushed agent then replans. If it has no legal
step, the requester stays put. The program asserts that the fixed point is reached rather than truncating the
iteration. A run \emph{fails} if some agent violates its own (M3), or if two consecutive steps pass with no
movement while some agent is incomplete and none is waiting by plan.

\subsection{Strict exact arrival}
\label{sec:strictfamilies}

Counts below are \emph{locations}, not cells; an instance on $\ell$ locations with horizon $H$ has
$\ell\cdot(H+1)$ cells per agent. There are three strict test cases---two new families, and the strict reading
of the instance that Section~\ref{sec:dwellstudy} uses under the other semantics---and all three are given as
generators in the accompanying program.

\begin{itemize}[leftmargin=1.6em,itemsep=3pt]
\item $S_{\mathrm{hold}}$---the same instance as $G_{\mathrm{give}}$ in Section~\ref{sec:dwellstudy}, read
under \textsc{strict}: a corridor $c_0-\cdots-c_5$ with a pocket $w$ adjacent only to $c_4$; seven locations,
two agents, horizon $6$. A through-traveller $P$ runs from $(c_0,0)$ to $(c_5,5)$ with zero slack;
a holder $C$ runs from $(w,0)$ to $(c_3,6)$ with slack four. Under \textsc{strict}, $C$ cannot sit on its goal
$c_3$ early, so it must wait somewhere; \emph{under the policy $\Pi$} it holds at the \emph{penultimate}
location $c_4$, which lies on $P$'s only path and which $P$ needs at $t=4$ with no slack. The discipline does
not require that choice---Section~\ref{sec:policyscope} exhibits strict schedules in which $C$ waits
elsewhere---so this is a counterexample to the policy, not to the departure rule alone.
\item $S_{\mathrm{late}}$: corridors $a_0-x-a_2-a_G$ and $b_0-b_1-x-b_3-b_G$ sharing exactly the location $x$;
eight locations, two agents, horizon $5$; $\mathrm{slack}_A=2$, $\mathrm{slack}_B=1$. Each agent has a private
corridor \emph{after} $x$ on which it can hold, and that is exactly what makes the instance feasible under
\textsc{strict}: waiting can be redistributed to a private location. Latest departure pushes all waiting to the
start, so both agents want $x$ at $t=3$ with no slack left; earliest departure separates them in time, $A$
crossing $x$ at $t=1$ and $B$ at $t=2$, each then holding privately.
\item $S_{\mathrm{late}}(r)$: designed against residual-$r$. A responder $B$ traverses
$b_0-x_1-\cdots-x_{r+1}-y-b_G$, of distance $r+3$, with slack $r+1$ and arrival time $T_B=2r+4$; the
\emph{private} location $y$ sits between the last shared location and the goal. For each $i$ a crosser $A_i$ occupies $x_i$ at exactly time $2i$ on a private
corridor of length $2i$, so its slack is zero and it departs at once under every departure rule. A $B$ that
departed by residual-$r$ meets $A_i$ at $x_i$ exactly when it has already been delayed $i-1$ times, is delayed
once at each of the $r+1$ shared locations, and misses its arrival; a $B$ that departed earliest is at $x_i$ at
time $i$, never at $2i$, and passes every crosser. The location $y$ is what the \textsc{dwell} design lacked. Deleting it while
keeping the slack at $r+1$---that is, shortening the arrival time to $T_B=2r+3$, which is exactly the
\textsc{dwell} family $G_{\mathrm{late}}(r)$ of Section~\ref{sec:dwellstudy}---forces $B$ to occupy $x_{r+1}$
at $T_B-1=2r+2$, where the zero-slack crosser $A_{r+1}$ also is, so that instance is infeasible under
\textsc{strict}. The deadline matters as much as the location: deleting $y$ but \emph{keeping} $T_B=2r+4$, and
so raising the slack to $r+2$, leaves a strict-feasible instance---the accompanying program enumerates it and
prints $6$, $26$ and $132$ schedules for $r=1,2,3$---which is a different family and not the one used below.
\end{itemize}

\begin{center}
\small
\begin{tabular}{@{}>{\raggedright\arraybackslash}p{0.19\linewidth}>{\raggedright\arraybackslash}p{0.19\linewidth}>{\raggedright\arraybackslash}p{0.54\linewidth}@{}}
\toprule
rule & defeated by & under which remaining policy parameters \\
\midrule
latest, i.e.\ residual-$0$ & $S_{\mathrm{late}}$ & all six push/tie combinations; no give-way step occurs, so
  the tie-break is never consulted \\
\addlinespace
residual-$r$, $1\le r\le8$ & $S_{\mathrm{late}}(r)$ & all six push/tie combinations, for the same reason; the
  crossers have zero slack and win every contest \\
\addlinespace
earliest & $S_{\mathrm{hold}}$ & under \textsc{yw} for both tie orders---and \textsc{yw} is vacuous under
  \textsc{strict}, so the holder simply cannot be moved. Under \textsc{yh} and \textsc{ya} it fails for the
  corner-first order and \emph{survives} the pocket-first order \\
\bottomrule
\end{tabular}
\end{center}

Three remarks. First, $S_{\mathrm{late}}$ also defeats residual-$1$, and $S_{\mathrm{hold}}$ also defeats
residual-$1$, residual-$2$ and residual-$3$; the table lists one defeating family per rule, not all of them.
Second, the coverage of the residual rules is \emph{independent of the push rule and of the tie order}, over
the whole grid: those failures are pure resource contests in which no give-way step occurs, so the tie-break is
never consulted. The same independence holds in the \textsc{dwell} study of Section~\ref{sec:dwellstudy}---it
is a property of the resource-contest mechanism, not of the arrival discipline. Third, the one rule whose
defeat does depend on the remaining parameters is earliest departure, in both semantics and for the same
reason: an early departer becomes an obstacle, and whether that is fatal depends on where it is pushed.

\begin{conjecture}[Coverage of the residual rules, strict exact arrival]
\label{conj:coverage}
For every $r\ge1$ the family $S_{\mathrm{late}}(r)$ defeats the rule residual-$r$ under every push rule and
every tie-break. Hence every departure rule that retains a fixed number of units of slack is defeated by some
member of the suite.
\end{conjecture}

The construction is parametric in $r$ and the argument above is uniform in $r$, but only $r=1,\dots,8$ has
been simulated, and only over the stated finite grid of push rules and two tie orders, so the general
statement is a conjecture and the table rows are claims about the listed values and settings.

\subsection{What the policy contributes, and what the discipline contributes}
\label{sec:policyscope}

It is tempting to read the strict results as facts about departure rules alone. They are not, and
$S_{\mathrm{hold}}$ shows why. Exhaustive enumeration gives that instance nine valid joint schedules under
\textsc{strict}, and in \emph{every} one of them the holder $C$ is in the pocket $w$ at $t=4$, the moment the
through-traveller needs $c_4$; $C$ enters its goal only at $t=6$. Eight of the nine have $C$ advance up the
corridor and then \emph{retreat} to the pocket---for instance $C$ visiting $c_4$ at $t=1$, returning to $w$ at
$t=4$, and going $w\to c_4\to c_3$ to arrive exactly on time; only the ninth waits in the pocket throughout.
Both counts are asserted by the accompanying program (block~S1).

So the strict discipline does \emph{not} force an early departer to sit where it obstructs. What forces that is
the policy: a shortest-path plan with no voluntary retreat, and holding at the penultimate location. Relax
either of those---let the plan back off into the pocket, or hold somewhere harmless---and $S_{\mathrm{hold}}$
stops being a counterexample. The give-way parameters are a third ingredient, but they decide only whether the
damage can be repaired once done: under \textsc{yw} no give-way step fires at all, so widening the give-way
from one step to many changes nothing there.

Two consequences. First, the results of Section~\ref{sec:strictfamilies} above, and of
Section~\ref{sec:dwellstudy} below, are properties of the policy class of Section~\ref{sec:policyclass}, and we
say ``defeated'' only in that sense. Second, the same caution applies to
the \textsc{dwell} study: there too the earliest rule's failure depends on where a pushed agent is sent. What
does \emph{not} depend on the push and tie choices in the tested grid is the residual rules' failure, which is
a resource contest in which no give-way step occurs at all---the tie-break is never consulted, and the outcome
is the same across every cell of the grid.

\subsection{Early arrival permitted: an independent robustness study}
\label{sec:dwellstudy}

The \textsc{dwell} discipline is not the one the theorems of this paper are about, but it is the one a vehicle
that may arrive early and wait on its stand actually obeys, so the same question is worth asking there. The
families are $G_{\mathrm{give}}$ (the corridor-and-pocket instance above, with $C$ camping on its goal rather
than holding before it), $G_{\mathrm{late}}$ (two corridors of unequal length sharing one location, both goals
adjacent to it) with a symmetric variant, and $G_{\mathrm{late}}(r)$, in which $B$ traverses
$b_0-x_1-\cdots-x_{r+1}-b_G$, of distance $r+2$, with slack $r+1$ and arrival time $T_B=2r+3$. The last is
obtained from $S_{\mathrm{late}}(r)$ by deleting the private location $y$ \emph{and} shortening $T_B$ by one so
as to preserve the slack; both changes are needed, as Section~\ref{sec:strictfamilies} explains. Feasibility under \textsc{dwell} is in the matrix above.

\begin{center}
\small
\begin{tabular}{@{}>{\raggedright\arraybackslash}p{0.19\linewidth}>{\raggedright\arraybackslash}p{0.19\linewidth}>{\raggedright\arraybackslash}p{0.54\linewidth}@{}}
\toprule
rule & defeated by & under which remaining policy parameters \\
\midrule
latest & $G_{\mathrm{late}}$ & all six push/tie combinations; no give-way step occurs \\
\addlinespace
residual-$r$, $1\le r\le8$ & $G_{\mathrm{late}}(r)$ & all six push/tie combinations \\
\addlinespace
earliest & $G_{\mathrm{give}}$ & under \textsc{yw} and \textsc{yh} for both tie orders; under \textsc{ya} for
  the corner-first order, \emph{surviving} the pocket-first order \\
\bottomrule
\end{tabular}
\end{center}

Two points come out of the sweep. On $G_{\mathrm{late}}$ the failure does not depend on the give-way details
at all---whoever loses the contest has no slack left, and no give-way step occurs---and in the symmetric
variant the agent index merely decides \emph{which} agent is lost. And the parameterized family is genuinely needed: residual-$1$ survives
both basic families, which is the gap $G_{\mathrm{late}}(1)$ closes.

\paragraph{A negative result, because it bounds the method.}
One would like to remove the last row's dependence on the tie-break by placing a mirrored second copy of
$G_{\mathrm{give}}$ alongside the first, so that any fixed choice fails on one of the two orientations. That
does not work, and cannot. The tie-break is a total order on \emph{locations}, and two copies occupy disjoint
locations, so a single order can prefer the pocket in both copies at once. The accompanying program builds a
genuinely mirrored pair---a second copy whose corridor is labelled in reverse---and confirms that this order
survives. \emph{The statement is about \textsc{ya}}: under \textsc{yw} a single copy already fails for both
tie orders, so there is nothing to remove. Under \textsc{ya}, no direct sum of copies can defeat every
location-order tie-break; doing so would require restricting the policy class to tie-breaks invariant under
relabelling, which we do not do.

\paragraph{Cost of the retreat to systematic search.}
Placing $m$ independent copies of $G_{\mathrm{give}}$ side by side and running the earliest rule with
\textsc{yw} and the ascending-index tie order, the greedy layer fails on every copy and the number of give-way
and block events grows linearly in $m$. As a model of the systematic layer we defined a naive joint
depth-first search whose children are ordered by \emph{increasing} free-arrival distance to the agent's own
goal---closest-to-goal first, the same rush-to-goal preference the reactive layer uses. \emph{This is our own
model, defined in the accompanying program; it is not an implementation of any published planner, and we have
measured none.} Because that ordering puts the rush-to-goal move first, the initial descent enters the trap in
every copy at once, and the node expansions to first solution grow $8,17,98,827$ for $m=1,2,3,4$.

\subsection{What this section does and does not show}
\label{sec:baselinescope}

It shows that a \emph{greedy layer} fails: a planner that computes free-arrival shortest routes, applies a
static departure offset, and resolves conflicts reactively by static priority with one-step give-way. It shows
nothing about complete search. Since every family is feasible in the semantics it is used for---which the
feasibility matrix establishes---any search that is complete for that model finds a solution on it. That is a
statement about the model. We have run no published planner, and make no claim whatever about the behaviour of
any implementation, in particular none about the LaCAM family \cite{lacam,lacamstar}, whose setting is standard
MAPF rather than either variant used here. The policy class is the one defined in
Section~\ref{sec:policyclass}; look-ahead policies, multi-step retreat, and policies that renegotiate planned
times lie outside it and we say nothing about them. Finally, the two semantics are kept apart on purpose: no
result of Section~\ref{sec:dwellstudy} is claimed for \textsc{strict}, and the main theorems of this paper
concern \textsc{strict} only.

\paragraph{What this section adds to Theorem~\ref{thm:A}.}
The two obstructions are orthogonal in mechanism and identical in root. In $G^\star$ the information that would
correct an early commitment is directionally unreachable; in $S_{\mathrm{hold}}$ and $S_{\mathrm{late}}$ it is
reachable but a static rule reads it wrongly. The common root is the \SUCC\ obligation isolated in step~(iii) of
Lemma~\ref{lem:A1}: a partially-ON branch is not a valid support, so whether an early cell may commit is
decided later and elsewhere.

\section{Implications for hardware, and the resource accounting}
\label{sec:hardware}

\subsection{The assumptions, stated in full}
\label{sec:hwassumptions}

A wall-clock reading of Corollary~\ref{cor:separation} rests on the following. Items (H1)--(H6) are load
bearing for the separation itself; (H7) is not, and we say below exactly what it does and does not buy.

\begin{enumerate}[leftmargin=1.8em,itemsep=1pt,label=(H\arabic*)]
\item bounded local degree;
\item \textbf{bounded per-cell state and bounded message alphabet}---note that Definition~\ref{def:class}(S)
does \emph{not} impose this, and both the gathering machine of Proposition~\ref{prop:nonvacuous} and the bus
machine of Section~\ref{sec:nonclaims}, item~\ref{nc:bus}, exploit its absence;
\item a fixed per-cell circuit, so that one cell's update is a bounded amount of work;
\item one hop per unit of physical time, so that a round is a unit of wall-clock time;
\item local wire length and delay independent of $N$---which a two-dimensional layout of a bounded-degree
fabric can provide but a fabric with long-range edges cannot;
\item clock distribution and synchronization across the array charged as a constant rather than as a function
of $N$;
\item the reduction tree and fan-out implementing a global aggregate flag charged as one round. A genuine
reduction over $N$ cells has depth $\Omega(\log N)$, so a fabric with global flags either pays that depth or
pipelines it.
\end{enumerate}

Under (H1)--(H6), on the family $\mathcal{H}_p$ with $N=\Theta(p)$ cells, a single-direction fabric with
$k=O(1)$ needs $\Omega(N)$ wall-clock time while a bidirectional one needs $O(1)$. The bidirectional side uses
no flags at all (Proposition~\ref{prop:U2}), and the one-directional side is a lower bound on rounds, so
(H7) is not needed for the separation---dropping it only makes each flag round cost more and the separation
wider. What (H7) buys is the reading of the \emph{matching upper bound}: without it,
Proposition~\ref{prop:U1}'s $\lceil p/k\rceil+O(1)$ rounds do not translate into $\Theta(N/k)$ wall-clock
time, and the tightness claim becomes a statement about rounds only.

Three further limits belong in the same breath. The separation is \emph{family-specific}: it is not a
statement about arbitrary multi-agent EA instances (Section~\ref{sec:nonclaims}, item~\ref{nc:notall}) and not
a lower bound against sequential algorithms (item~\ref{nc:cpu}). The upper bound of
Proposition~\ref{prop:U1} lives in the formal class, which permits topology-specific addresses and growing
state, so it is \emph{not} an upper bound for a machine constrained by (H1)--(H7); and
Proposition~\ref{prop:nonvacuous} lives there too, so the fact that the class can solve the problem at all is
not by itself a hardware claim. Finally, (H1)--(H7) constrain resources per cell and per wire; they do not by
themselves say the fabric is one tile repeated. A repeated-tile design is the engineering motivation of
Section~\ref{sec:whycellular}, not a clause of the formal model, and nothing below is derived from it.

\subsection{Five currencies}
\label{sec:currencies}

The results assemble into what we state as a \emph{scoped accounting principle}, not a theorem: within the
models analyzed here, two-sided coordination information is never free. It is paid in one of five currencies.

\begin{center}
\small
\begin{tabular}{@{}>{\raggedright\arraybackslash}p{0.24\linewidth}>{\raggedright\arraybackslash}p{0.30\linewidth}>{\raggedright\arraybackslash}p{0.38\linewidth}@{}}
\toprule
currency & where it is charged & reference \\
\midrule
preprocessing & the cavity, being forward reachability from $s$ intersected with backward reachability to $g$,
  is a two-sided computation performed off the machine
  & Lemma~\ref{lem:cavity-max}, Cor.~\ref{cor:t1-vacuous} \\
rule structure & every minimal sufficient rule basis imposes both \PRED\ and \SUCC
  & companion paper \cite{paperS} \\
local propagation & causally effective dependence in both temporal orientations
  & Theorem~\ref{thm:A}, Cor.~\ref{cor:composite} \\
global bandwidth & for a one-directional machine, $J_R+kR\ge p$, hence $\lceil p/k\rceil$ rounds when $J_R=0$
  & Theorem~\ref{thm:Bprime}, Prop.~\ref{prop:U1} \\
joint-state representation & a state space exponential in the number of agents
  & Proposition~\ref{prop:jointlift} \\
\bottomrule
\end{tabular}
\end{center}

These are not all of one kind: the rule-structure obligation constrains \emph{which states are stable}, whereas
propagation, bandwidth and representation are \emph{runtime} resources. The principle is an accounting across
two ledgers, not a single conserved quantity, and we do not claim the list is exhaustive.

\section{Related work}
\label{sec:related}

We state relations, not priority. Our search is a first pass at the level of bibliographic records and
abstracts (Section~\ref{sec:searchscope}); we did not find prior work on the exact conjunction studied here,
and that is a scoped statement, not a claim of nonexistence.

\subsection{One-way cellular automata}
\label{sec:onewayca}

Restricting a cellular array to one-way communication is an established subject with a continuous literature
from 1980 onward. Of the neighbourhoods we surveyed it is the one that shares our question most directly. One-way bounded cellular automata were
introduced by Dyer \cite{dyer1980}; Umeo, Morita and Sugata gave a deterministic one-way simulation of two-way
real-time cellular automata \cite{umeo1982}; Ibarra and Jiang \cite{ibarra1987} and Chang, Ibarra and Vergis
\cite{chang1988} developed the power of one-way arrays and of one-way communication; R\'oka extended one-way
cellular automata to Cayley graphs \cite{roka1994}; Terrier exhibited a language not recognizable in real time
by one-way cellular automata \cite{terrier1996}; Kutrib and Malcher studied arrays whose communication is
restricted in the number or the width of messages \cite{kutrib2009,kutrib2010}; and Grandjean, Grente and
Terrier gave a recent logical account of what real-time one-way and two-way arrays compute
\cite{grandjean2024}. We have verified the
bibliographic record of every one of these against primary publisher metadata; we have \emph{not} read the full
texts (Section~\ref{sec:searchscope}).

That lineage and this paper share one question---what does forbidding one direction of communication
cost?---and differ in nearly every other coordinate. The differences are not defensive caveats; they are what
the present model is for.

\begin{center}
\small
\begin{tabular}{@{}>{\raggedright\arraybackslash}p{0.28\linewidth}>{\raggedright\arraybackslash}p{0.30\linewidth}>{\raggedright\arraybackslash}p{0.36\linewidth}@{}}
\toprule
lineage & what is studied there & what is different here \\
\midrule
one-way bounded and cellular arrays \cite{dyer1980,ibarra1987} & language recognition on a uniform array; a
  single acceptor output & runtime endpoint queries on a fixed nonuniform topology; a distributed per-cell mask
  output \\
\addlinespace
one-way simulation of two-way arrays \cite{umeo1982} & whether, and with what time overhead, one-way
  simulates two-way & target-local \emph{actual} dependence, and forcing over a query family rather than
  simulation \\
\addlinespace
the power of one-way communication \cite{chang1988} & comparing the power of one-way and two-way parallel
  models & an aggregate transcript, per-cell outputs, and direct-sum forced bits \\
\addlinespace
one-way automata on other graphs \cite{roka1994} & one-way simulation on a general communication graph &
  agent-ranked DAGs, time-neutral conflict edges, and endpoint-conditioned input \\
\addlinespace
real-time non-recognizability \cite{terrier1996,grandjean2024} & a language a one-way array cannot recognize
  in real time, and the logical characterisation of what real-time arrays compute & a forced disagreement at
  one named cell on one fixed topology, under a promise that the query is feasible \\
\addlinespace
sparse or one-bit communication \cite{kutrib2009,kutrib2010} & limiting the number or the width of local
  messages & a finite local alphabet whose width is not bounded by a constant across topologies, combined
  with an aggregate-only global channel and an orientation restriction \\
\bottomrule
\end{tabular}
\end{center}

Three differences carry the weight. \emph{Nonuniformity}: our machines may compile the topology and each cell's
address into the transition rule (Remark~\ref{rem:nonuniform}), whereas that literature studies uniform
arrays; nonuniformity \emph{strengthens} our lower bounds for the task posed here, since a
uniform-code machine performing this task is a special case. It is not the reason the results fail to transfer
to that literature's setting; the reasons are the next two, together with the topology and output conventions
(Section~\ref{sec:nonclaims}, item~\ref{nc:nonuniform}). \emph{The output task}: a per-cell membership bit of a
valid joint mask rather than one acceptance bit, so the obstruction is a forced disagreement at a named cell
rather than a language separation. \emph{The input}: runtime endpoint markers together with the exact
reachability cache, on a topology fixed offline---a pairing we did not find in the sources surveyed, and the
reason Corollary~\ref{cor:t1-vacuous} is the pivot of the paper. Until a full-text pass is done we claim only
that this exact conjunction was not found in those sources.

\subsection{Other neighbours}
\label{sec:neighbours}

\paragraph{Transcript counting.}
The counting in Lemma~\ref{lem:F1prime} is standard: bound the number of interaction transcripts against the
number of outputs that must be distinguished, in the tradition of communication as a resource \cite{yao1979},
round-restricted communication complexity \cite{nisanwigderson1993}, information crossing a cut
\cite{pelegrubinovich2000}, and the number-in-hand and shared-blackboard settings
\cite{kushilevitznisan1997,phillips2016,hazankushilevitz2017}. We claim no new counting principle, and no priority for the
pairing either: what we did not find in the sources surveyed is its combination with a directed-cone locality
statement on this task and model.

\paragraph{Locally connected arrays with global buses.}
The architecture pattern is classical: meshes with a broadcast facility \cite{stout1983}, array processors with
a global bus \cite{bokhari1984}, meshes with multiple broadcasting \cite{chen1990}, and VLSI area--time and
bused-mesh lower bounds \cite{thompson1979,condon1996,cheunglau1993}; weaker and stronger global primitives
include the beeping model \cite{cornejokuhn2010} and the Broadcast Congested Clique
\cite{druckerkuhnoshman2014}. Against these, our global primitive is a many-to-one commutative aggregate
returned to all cells; the runtime input is an endpoint query on an offline-known topology together with its
reachability cache; the output is a per-cell membership bit; and the lower-bound object is a direct-sum family
forcing target bits through the transcript.

\paragraph{Locality and advice.}
The locality theory of distributed graph algorithms \cite{linial1992,naorstockmeyer1995} and local distributed
decision \cite{fraigniaud2011} is the natural home of Lemma~\ref{lem:cone}. Distributed computation with advice
\cite{fraigniaud2009advice,fraigniaud2010advice,balliu2025} is the closest relative of $J_R$, the differences
being that advice there is designed whereas the cavity is fixed by the problem structure, and that on our
gadget the relevant cones are provably query-invariant, so $J_R=0$ is a theorem rather than a modelling choice.
Reachability and distance labels \cite{cohen2003} are conceptually adjacent to the cavity, and proof-labeling
schemes \cite{korman2010,feuilloley2016} distribute per-node certificates for local verification; both answer
queries centrally or verify a given assignment, whereas our task is distributed computation of the bits under a
direction restriction and conflict coupling.

\paragraph{Program chopping.}
A program chop collects the elements through which effects propagate from a source to a target
\cite{jacksonrollins1994}, and interprocedurally the naive intersection of a forward and a backward slice
over-approximates it \cite{repsrosay1995}. On a plain digraph the chop is exactly the support region, so
computing all support elements between two endpoints is not new; the difference is operational, chopping being
a centralized analysis.

\paragraph{Broadcast CONGEST reachability.}
Among the lower bounds we surveyed, the one whose proof shape is nearest to ours is the Broadcast CONGEST
single-source reachability bound of Chechik and Mukhtar \cite{chechikmukhtar2019}, proved by the same
locality-plus-transcript-counting argument; see
also \cite{ghaffariudwani2015,fanwangwu2012}. Their family varies the graph and each vertex reports
single-source reachability under per-node broadcasts; ours varies only runtime endpoint markers on a fixed
topology, supplies the reachability cache for free, targets membership bits of \emph{any} valid joint mask, and
confines the global channel to one-bit many-to-one aggregates.

\paragraph{Multi-agent path finding.}
The setting is extended MAPF with pinned departure and arrival \cite{stern2019,mapftw,precedence}; the LaCAM
lineage is the modern reference point for large-scale MAPF \cite{lacam,lacamstar} and we do not claim to
surpass it---Section~\ref{sec:baseline} is explicitly not about it. Representations of the full solution set
have been studied \cite{bardugo2025}. Single-agent planning over an explicit space-time graph is classical
\cite{sipp}, and cellular-automata planners propagate over space-time grids \cite{marchese2006}; earlier work
by the author and colleagues gave a SAT-based implementation of a MAPF variant on reconfigurable hardware
\cite{huang2022}. Our aim differs: we prove what a restricted machine class cannot do.

\subsection{Scope of the search}
\label{sec:searchscope}

We surveyed, at the level of bibliographic records and abstracts rather than full-text reading: one-way
cellular automata and one-way communication in parallel models; locality and local distributed decision;
distributed computation with advice; VLSI, mesh and bused-mesh lower bounds; multiparty communication
complexity and shared blackboards; and reachability, distance and proof-labeling schemes. Within that scope we
did not find prior work studying direction-restricted distributed lower bounds under per-cell
reachability-cache inputs with per-cell mask output. We have \emph{not} read the full texts of the one-way
cellular automata literature, nor covered the citation graph, exhaustive proceedings scans, temporal networks,
self-stabilization, population protocols, distributed runtime verification, non-English literature, theses,
patents, or incompletely indexed 2025--2026 preprints. This is a scoped statement, not a claim of exhaustive
priority, and the full-text pass over Section~\ref{sec:onewayca} remains open (Section~\ref{sec:future}).

\section{Discussion}
\label{sec:discussion}

\subsection{What is not claimed}
\label{sec:nonclaims}

\begin{enumerate}[leftmargin=1.8em,itemsep=4pt]
\item\label{nc:bus} \textbf{``Forward-only machines cannot solve this'' is false without qualification.} With
$k\ge1$ and unbounded rounds there is a forward-only machine that solves any instance of the class. Since the
topology is known offline, give each cell a static address and schedule one designated cell per round to
contribute a bit to an OR flag; then run a \emph{fixed-length schedule determined by the topology alone}: (i)
collect the endpoint markers and cavity bits in address order, (ii) use the factored cell bits as a distributed
tape enumerating \emph{all} candidate joint masks, (iii) for each candidate let every cell evaluate its local
path condition, pin condition and slot exclusion and aggregate any violation through the flag, each cell
retaining the lexicographically first candidate found valid so far, and (iv) at the topology-determined
worst-case round $R_*$---the same for every query---have every cell emit its bit of the retained candidate.
Termination is at $R_*$ whatever the query, as clause (T) requires; the machine does not stop early when it
finds a solution. This machine uses no temporal local channels at all, so by clause (D) it is forward-only. What it pays
is an enormous $R$ and an enormous per-cell state count, which clause (S) permits. \emph{This is why the main
results are stated as a pair}: an impossibility at $k=0$ (Theorem~\ref{thm:A}) and a price at $k>0$
(Theorem~\ref{thm:Bprime}).
\item \textbf{We retract ``forward-only forces joint-state expansion.''} The trade-off admits several
escapes---a long run (item~\ref{nc:bus}), a wide channel $k=\Omega(p)$, external advice $a\ge p$
(Remark~\ref{rem:advice}), sequential enumeration over the factored variables, or shared memory. What follows
from Theorem~\ref{thm:Bprime} is the disjunction ``pay $J_R+kR\ge p$, or admit past-directed local
propagation.'' Proposition~\ref{prop:jointlift} is an \emph{upper} bound and cannot serve as the opposing side
of a representation lower bound.
\item\label{nc:notall} \textbf{``Forward-only fails on every multi-agent family'' is false.} Symmetric families
solvable by static priority exist: in the two-diamond witness of the companion paper \cite{paperS}, fixing a
priority resolves the conflict with zero rounds of information. Vertex conflicts plus cavity input do not by
themselves imply a direction requirement. Our theorems are of existential-topology, universal-query form
(Remark~\ref{rem:quant}).
\item\label{nc:samequery} \textbf{No same-query bidirectionality.} Corollaries~\ref{cor:composite}
and~\ref{cor:separation} are architecture-level and query-family-level. We do not show that on a single query,
or along a single causal path, both orientations are simultaneously causally effective.
\item\label{nc:feasible} \textbf{Feasibility promise.} All results assume the query is feasible. Detecting
infeasibility, completeness for $\mathrm{FAIL}$, and certificates of infeasibility are not studied, and the
lower bounds say nothing about them.
\item\label{nc:ea} \textbf{The direction lower bounds do not separate exact arrival from the free-arrival or
deadline variants.} Theorems~\ref{thm:A} and~\ref{thm:Bprime} and the propositions around them are proved for
exact-arrival instances, where the obstruction is operationally natural, but no proof uses the illegality of
early arrival, and we do not show that analogous gadgets fail to port to those variants. This is a statement
about Sections~\ref{sec:warmup}--\ref{sec:bandwidth}. Section~\ref{sec:baseline} does compare two arrival
disciplines, and there the difference is real---families feasible under one are infeasible under the
other---but that comparison is about static departure policies, not about the direction bounds.
\item\label{nc:conflicts} \textbf{Vertex conflicts only.} Edge conflicts, swap constraints in the
machine-class part, capacities greater than one, kinematic attributes and asynchrony are outside the model.
(The enumeration of Section~\ref{sec:baseline} does forbid swaps; the machine class of Section~\ref{sec:class}
does not model them.)
\item\label{nc:nonuniform} \textbf{Not a statement about classical cellular automata.} The lower-bound class
permits topology-specific transition code and per-cell state; since that only adds power, the bounds do cover
uniform-code machines \emph{performing the task defined here}. They say nothing about the classical
one-way cellular automaton setting, which differs in the task (language recognition with a single acceptor
output), in the input convention (no runtime endpoint query, no reachability cache) and in the topology
(a uniform array rather than agent-ranked DAGs with conflict edges). Transferring anything to that setting
needs a separate argument (Section~\ref{sec:onewayca}).
\item \textbf{Nothing about NP-hardness or exponential time.} Bidirectionality does not move a problem between
complexity classes.
\item\label{nc:cpu} \textbf{No claim that a CPU is necessarily $N$ times slower.} Under (H1)--(H7) a sequential
machine \emph{simulating the same round update} pays $\Theta(N)$ work per round; that is a statement about one
simulation strategy, not a lower bound against all sequential algorithms.
\item \textbf{No claim about any deployed system.} The results bound a machine class. We prove nothing about
the convergence of any stochastic dynamics, and assume none.
\item \textbf{Not a complete lower bound for deadlock detection.} It is a direct-sum lower bound for the forced
joint-extendability bit of a constant-size selector.
\item \textbf{Randomization is out of scope.} All machines here are deterministic; a bounded-error extension
would need a distributional form of the transcript counting.
\end{enumerate}

\subsection{Status of each claim}
\label{sec:status}

We separate three questions that are easy to conflate: whether a statement is proved, whether the proof has
been mechanized, and whether its priority has been established.

\begin{center}
\small
\begin{tabular}{@{}>{\raggedright\arraybackslash}p{0.25\linewidth}>{\raggedright\arraybackslash}p{0.19\linewidth}>{\raggedright\arraybackslash}p{0.28\linewidth}>{\raggedright\arraybackslash}p{0.20\linewidth}@{}}
\toprule
statement & mathematical status & mechanization status & priority status \\
\midrule
Lemma~\ref{lem:cavity-max}, Prop.~\ref{prop:jointlift}, Prop.~\ref{prop:nonvacuous}, Lemma~\ref{lem:cone},
Lemma~\ref{lem:F1prime}, Lemma~\ref{lem:reversal}
  & proved here & not formalized & standard or elementary; none claimed \\
\addlinespace
Lemma~\ref{lem:A1}, Lemma~\ref{lem:A2}, and their time-reversed forms
  & proved here & exhaustively enumerated for $L=1$ to $6$, both output conventions; the $L=1$ case is also
    checked in Lean under the finite Boolean encoding described below
  & first pass only \\
\addlinespace
Lemma~\ref{lem:A1sum}
  & proved here, uniformly in $p$ & its forcing clause exhaustively enumerated for $p=1,2,3$ over all $2^p$
    queries; its cone clause follows from the single-copy check of Lemma~\ref{lem:A2}
  & first pass only \\
\addlinespace
Theorem~\ref{thm:A} and Cor.~\ref{cor:composite}
  & proved here & the pair-invariance contradiction form, and the two-gadget corollary in that form, are
    checked in Lean & first pass only; full-text pass over Section~\ref{sec:onewayca} pending \\
\addlinespace
Theorem~\ref{thm:Bprime}, Cor.~\ref{cor:separation}
  & proved here & not formalized: the transcript counting is not in the Lean development
  & first pass only \\
\addlinespace
Prop.~\ref{prop:thm5}, Prop.~\ref{prop:U1}, Prop.~\ref{prop:U2}
  & proved here & not formalized & standard; none claimed \\
\addlinespace
Theorem~\ref{thm:single}, Cor.~\ref{cor:single-trivial}
  & proved here & not formalized; the fork of Theorem~\ref{thm:single} is exercised by the enumeration of
    Cor.~\ref{cor:t1-vacuous} & standard; none claimed \\
\addlinespace
Cor.~\ref{cor:t1-vacuous}, Remark~\ref{rem:zplacement}
  & proved here & exhaustively enumerated & --- \\
\addlinespace
Feasibility of the basic families of Section~\ref{sec:baseline}, both semantics
  & established by the accompanying program, not by a paper proof
  & exhaustively enumerated: every valid joint schedule generated and counted
  & --- \\
\addlinespace
Feasibility of $S_{\mathrm{late}}(r)$ under both semantics, and of $G_{\mathrm{late}}(r)$ under \textsc{dwell}
  & as above & by an explicit witness schedule verified against (M1)--(M5); the
    full schedule sets are not generated & --- \\
\addlinespace
Strict infeasibility of $G_{\mathrm{late}}(r)$
  & paper argument: a forced collision at $T_B-1$ & confirmed by enumeration for
    $r=1,2,3$ only; the argument, not the enumeration, covers general $r$ & --- \\
\addlinespace
The two coverage tables of Section~\ref{sec:baseline}
  & as above & policy simulation over the stated finite grid---three push rules
    against two chosen tie orders---for $r=1$ to $8$ & --- \\
\addlinespace
Conjecture~\ref{conj:coverage}
  & \emph{conjecture} & policy-simulated over that grid for $r=1$ to $8$ only
  & --- \\
\bottomrule
\end{tabular}
\end{center}

We use \emph{machine-checked} only for the Lean development and \emph{exhaustively enumerated} only for the
Python programs, and neither phrase for a paper proof.

\paragraph{What the Lean development checks.}
\label{sec:leanscope}
\emph{The finite $L=1$ gadget lemmas and the pair-invariance contradiction forms have been checked in Lean~4;
the operational cone lemma remains a paper lemma.} Concretely: the twelve-cell type with its edge, rank and
slot maps; that the two conflicts are crossed and lie in one rank; that the reversed edge relation is the
reversal (this one holds by definition, and the file records it as such); uniqueness of the valid joint output and the forcing of the target bit \emph{for the finite Boolean
validity encoding used in that file}; the fixpoint computation of both cones and the equality of the labelled
cone views across the two queries; the four contradiction forms; and the two composite-family corollaries.
\verb|#print axioms| reports \verb|propext| only, with no \verb|sorryAx| and no user-declared axiom.

Three things lie outside it, and we state them because the wording is easy to inflate. First, validity is
written in the Lean file as hand-written gadget-specific Boolean predicates; that these agree with the generic
path-support and slot-exclusion predicates induced by the edge and slot relations is \emph{not} proved inside
Lean---it is a paper argument, cross-checked by the Python enumeration, which builds the predicates generically
from the edge and slot relations. Second, the reversed side sets its validity predicate equal to the forward
one rather than re-deriving it from the reversed edge relation; the duality is Lemma~\ref{lem:reversal}, a
paper lemma. Third, pair-invariance is supplied as an explicit hypothesis, and the transcript
counting---hence Theorem~\ref{thm:Bprime}---is not formalized at all. Accordingly we do not write, and ask that
others not write, that Theorem~\ref{thm:A} is ``Lean verified'' without that qualification.

\subsection{Future work}
\label{sec:future}

Experiments are not reported here and no claim above depends on any. The open items are, in order of how much
they would change the paper: the full-text pass over the one-way cellular automata literature of
Section~\ref{sec:onewayca}, which could reclassify our theorem as a specialization, a corollary, or a genuinely
different output-task separation; formalizing the semantic bridge and then Lemma~\ref{lem:cone}, which would let the
mechanization claim be stated without its present caveats; a randomized extension, needing a distributional
transcript count; query-dependent termination, which changes the exact constants because the halting time
itself becomes a channel; edge conflicts, capacities and kinematics; and upper bounds for bidirectional
machines on general topologies, sharper than the gathering bound of Proposition~\ref{prop:nonvacuous}. We also
note that the families of Sections~\ref{sec:mainA} and~\ref{sec:baseline} are small, parameterized and come
with checked ground truth---enumerated for the basic families, a verified witness for the parameterized
ones---which makes them usable as a benchmark suite---a use we have not exercised and
about which we claim nothing.

\section{Conclusion}
\label{sec:conclusion}

In exact-arrival routing an early commitment can be wrong for a reason that lies both later in time and in
another agent. We asked what that costs a cellular machine---one automaton per cell, factored state, local
channels, synchronous rounds---and answered under two concessions that remove the easy objections. Handing the
machine every agent's exact reachability cache makes the single-agent direction theorem vacuous, and handing it
$k$ one-bit global aggregates makes an unconditional impossibility false. Between those two concessions the
result survives: on a twelve-cell, two-agent, depth-three topology, no forward-only machine with no global
channel is correct on two feasible queries, at any number of rounds; the time-reversed topology does the same
to backward-only machines; and on the disjoint union of the two, one-directional causal dependence is excluded
outright. The class is not vacuous---a flooding machine solves every instance in as many rounds as the diameter of the
local fabric---so what the one-directional machines lack is an orientation, not expressive power. With $k$ flags the impossibility becomes
the price $J_R+kR\ge p$ at constant depth, tight in order, against a constant number of rounds for a machine
that may use the other orientation. The mechanism throughout is a single move: put the forcing runtime bit in
the collision partner rather than in the target's own agent, so that the target's own cache is constant while
the correct answer is not.


\end{document}